\documentclass[a4paper,11pt]{article}

\usepackage[latin1]{inputenc}
\usepackage{amsmath,amssymb}
\usepackage{amsthm}
\usepackage{a4wide}
\usepackage{graphicx}
\usepackage{enumerate}
\usepackage{float}
\usepackage{calc}

\usepackage[round]{natbib}
\newtheorem{theorem}{Theorem}[section]
\newtheorem{lemma}[theorem]{Lemma}

\newtheorem{corollary}[theorem]{Corollary}
\newtheorem{proposition}[theorem]{Proposition}

\floatstyle{ruled}
\newfloat{algorithm}{htbp!}{loa}
\floatname{algorithm}{Algorithm}

\newlength{\algindent}
\newenvironment{AlgorithmSteps}{%
  \setlength{\algindent}{0mm}
  \begin{enumerate}[\bf\small 1.]%
  \setlength{\itemsep}{3pt} 
}{%
  \end{enumerate}

}
\newcommand{\Step}[1]{\item\hspace*{\algindent}\parbox[t]{\textwidth-\algindent-\leftmargin}{#1}}
\newcommand{\Comment}[1]{\item[]\hspace*{\algindent}\parbox[t]{\textwidth-\algindent-\leftmargin}{\sl//#1}}
\newcommand{\IncreaseIndent}{\addtolength{\algindent}{7mm}}
\newcommand{\DecreaseIndent}{\addtolength{\algindent}{-7mm}}

\newcommand{\A}{\mathcal{A}}
\newcommand{\C}{\mathcal{C}}
\newcommand{\MH}{\mathcal{H}}
\newcommand{\I}{\mathcal{I}}
\newcommand{\M}{\mathcal{M}}
\newcommand{\Q}{\mathcal{Q}}
\newcommand{\eg}{\textit{e.g.}}
\newcommand{\ie}{\textit{i.e.}}
\newcommand{\lin}{\ensuremath{\text{\sc Lin}}}
\newcommand{\cir}{\ensuremath{\text{\sc Cir}}}
\renewcommand{\to}{\ensuremath{\longrightarrow}}

\newlength{\lengthTo}
\newlength{\lengthNot}
\newcommand{\nto}{%
 \ensuremath{%
  \to\hspace*{-.5\lengthTo}\hspace*{-\lengthNot}\not\ \hspace*{.5\lengthTo}%
 }%
}

\title{Subclasses of Normal Helly Circular-Arc Graphs}
\author{%
  Min Chih Lin~\thanks{CONICET}~$^\dag$ \and%
  Francisco J.\ Soulignac~\thanks{Universidad de Buenos Aires, Facultad de Ciencias Exactas y Naturales, Departamento  de  Computaci\'on, Buenos Aires, Argentina. {\small \{\texttt{oscarlin, fsoulign}\}\texttt{.dc.uba.ar}}} \and%
  Jayme L.\ Szwarcfiter~\thanks{Universidade Federal do Rio de Janeiro, Instituto de Matem\'atica, NCE and COPPE, Caixa Postal 2324, 20001-970 Rio de Janeiro, RJ, Brasil. {\small \texttt{jayme@nce.ufrj.br}}}%
}

\date{}

\begin{document}
\maketitle
\begin{abstract}
 A Helly circular-arc model $\M =(C,\A)$ is a circle $C$ together with a Helly family $\A$ of arcs of $C$. If no arc is contained in any other, then $\M$ is a proper Helly circular-arc model, if every arc has the same length, then $\M$ is a unit Helly circular-arc model, and if there are no two arcs covering the circle, then $\M$ is a normal Helly circular-arc model. A Helly (resp.\ proper Helly, unit Helly, normal Helly) circular-arc graph is the intersection graph of the arcs of a Helly (resp.\ proper Helly, unit Helly, normal Helly) circular-arc model. In this article we study these subclasses of Helly circular-arc graphs.  We show natural generalizations of several properties of (proper) interval graphs that hold for some of these Helly circular-arc subclasses.  Next, we describe characterizations for the subclasses of Helly circular-arc graphs, including forbidden induced subgraphs characterizations. These characterizations lead to efficient algorithms for recognizing graphs within these classes.  Finally, we show how do these classes of graphs relate with straight and round digraphs.

 \vspace*{3mm} {\bf Keywords:} Helly circular-arc graphs, proper circular-arc graphs, unit circular-arc graphs, normal circular-arc graphs
\end{abstract}

\section{Introduction}
\label{sec:introduction}

A graph is an \emph{interval graph} if its vertices can be arranged into a one-to-one mapping with a family of intervals of the real line in such a way that two vertices of the graph are adjacent if and only if their corresponding intervals have nonempty intersection.  The family of intervals is a called an \emph{interval model} or \emph{interval representation} of the graph.  Similarly, a graph is a \emph{circular-arc graph} if its vertices can be arranged into a one-to-one mapping with a family of arcs of some circle in such a way that two vertices of the graph are adjacent if and only if their corresponding arcs have nonempty intersection.  The circle together with the family of arcs is a called a \emph{circular-arc model} or \emph{circular-arc representation} of the graph.

Circular-arc graphs are a natural generalization of interval graphs; the intersection of segments of a line (\ie\ intervals) is replaced with the intersection of segments of a circle (\ie\ arcs).  Despite the replacement of a line with a circle can seem like a petty modification, circular-arc graphs have a much more complex structure than interval graphs.  One of the main reasons behind this complexity is that arcs can intersect in different ways as intervals do.  In an interval model, each family $\I$ of pairwise intersecting intervals is associated with a unique nonempty segment $(s, t)$ of the real line such that $(s, t)$ is contained in every interval of $\I$, and $(s, t)$ is inclusion-maximal.  The situation for circular-arc models is quite different, since for some families of pairwise intersecting arcs there could be zero or more than one of such segments. As a result of this generality, many properties that hold for interval graphs cannot be naturally generalized to circular-arc graphs.  Let $\C$ be the class of circular-arc models such that every family $\A$ of pairwise intersecting arcs is associated with a unique nonempty segment $(s,t)$ of the circle such that $(s,t)$ is contained in every arc of $\A$, and $(s,t)$ is inclusion-maximal.  It is not hard to see that a circular-arc model $\M$ belongs to $\C$ if and only if $\M$ has no two nor three arcs that together cover its circle.

The circular-arc models in $\C$ are related with two important subclasses of circular-arc models, namely normal and Helly circular-arc models.  A circular-arc model is \emph{normal} when no two arcs cover the circle, while it is \emph{Helly} when every family of pairwise intersecting arcs share a common point of the circle.  \citet{LinSzwarcfiter2006} proved that if a circular-arc model is not Helly, then it contains two or three arcs that together cover the circle.  Furthermore, if a normal circular-arc model has three arcs that cover the circle, then the model is not Helly. Consequently, $\C$ is precisely the class of circular-arc models that are simultaneously normal and Helly.  We define the \emph{normal Helly} circular-arc graphs as those circular-arc graphs that admit a model in $\C$.  

The purpose of this paper is to study the class of normal Helly circular-arc graphs, as well as its subclasses.  In Section~\ref{sec:preliminaries} we give the basic definitions used throughout the paper, and we introduce the families of circular-arc models and graphs.  The reasons for studying these subclasses of circular-arc graphs are given in Section~\ref{sec:motivation}, where we compile several properties of interval graphs that hold for normal Helly circular-arc graphs but do not hold for the general classes of normal circular-arc and Helly circular-arc graphs.  Of course, there is also a completely theoretical motivation for studying normal Helly circular-arc graphs, namely to study all the circular-arc subclasses that are obtained by the intersection of the most well known subclasses of circular-arc models.  Next, in Section~\ref{sec:characterizations}, we describe the forbidden induced subgraph characterizations of the normal Helly subclasses, while in Section~\ref{sec:algorithms} we deal with the corresponding recognition problems.  Finally, in Section~\ref{sec:properties} we provide some additional results on normal Helly circular-arc graphs that further relate this class of graphs with the class of interval graphs, and in Section~\ref{sec:conclusions} we discuss some possibilities for future research. 

\section{Preliminaries}
\label{sec:preliminaries}

For a graph $G$, we use $V(G)$ and $E(G)$ to denote the sets of vertices and edges of $G$, respectively, while we use $n$ and $m$ to denote $|V(G)|$ and $|E(G)|$, respectively.  We will write $uv$ to represent the edge of $G$ between the pair of adjacent vertices $u$ and $v$. For $v \in V(G)$, denote by $N_G(v)$ the set of vertices adjacent to $v$, and $N_G[v] = N_G(v) \cup \{v\}$.  The \emph{degree} of $v$ is $d_G(v) = |N_G(v)|$, and $v$ is \emph{universal} if $N_G[v]=V(G)$. Two vertices $v$ and $w$ are \emph{twins} in $G$ if $N_G[v] = N_G[w]$.  We omit the subscript from $N$ and $d$ when there is no ambiguity about the referred graph.

The subgraph of $G$ \emph{induced} by $V \subseteq V(G)$, denoted by $G[V]$, is the graph that has $V$ as vertex set and two vertices of $G[V]$ are adjacent if and only if they are adjacent in $G$.  A \emph{clique} is a maximal subset of pairwise adjacent vertices.  A \emph{hole} is a chordless cycle with at least four vertices.  A \emph{bipartite} graph is a graph whose vertex set can be partitioned into two sets of pairwise non-adjacent vertices, where one of the partitions may be empty.  A \emph{co-bipartite} graph is the complement of a bipartite graph.

As for graphs, we respectively use $V(D)$ and $E(D)$ to refer to the sets of vertices and edges of a digraph $D$, while we use $n$ and $m$ to denote $|V(D)|$ and $|E(D)|$, respectively.  We write $v \to w$ to indicate that there is a directed edge from $v$ to $w$ in $D$, and $v \nto w$ to indicate that there is no such directed edge.  The \emph{underlying graph} of $D$ is the graph $G(D)$ with vertex set $V(D)$ such that $v$ is adjacent to $w$ in $G$ if and only if $v$ is adjacent to $w$ in $D$.  When $v \to w$, we say that $v$ is an \emph{in-neighbor} of $w$ and that $w$ is an \emph{out-neighbor} of $v$.  The \emph{inset} of $v$ is the set $N_D^-(v)$ of all the in-neighbors of $v$, and the \emph{outset} of $v$ is the set $N_D^+(v)$ of all the out-neighbors of $v$.  The \emph{closed inset} is $N_D^-[v] = N_D^-(v) \cup \{v\}$, and the \emph{closed outset} is $N_D^+[v] = N_D^+(v) \cup \{v\}$. The cardinality of $N_D^-(v)$, denoted by $d_D^-(v)$,  is the \emph{indegree} of $v$ and the cardinality of $N_D^+(v)$, denoted by $d_D^+(v)$, is the \emph{outdegree} of $v$.  As for graphs, we omit the subscripts in $N$ and $d$ when there is no ambiguity about the referred digraph. 

The subdigraph of $D$ \emph{induced} by $V \subseteq V(D)$, denoted by $D[V]$, is the digraph that has $V$ as vertex set and $v \to w$ in $D[V]$ if and only if $v \to w$ in $D$.  A digraph $D$ is an \emph{oriented graph} when either $v \nto w$ or $w \nto v$, for every $v,w \in V(D)$.  In other words, $D$ is an oriented graph if it can be obtained from a graph $G$ by choosing an orientation for each edge $vw$ of $G$.  In such case, we call $D$ an \emph{orientation} of $G$, \ie, $D$ is an \emph{orientation} of $G$ if $D$ is an oriented graph whose underlying graph is isomorphic to $G$.

In this paper we mainly deal with collections that are of a circular (cyclic) nature, such as circles (viewed as a collection of points), circular orderings, etc. Usually, the objects of a collection are labeled with some kind of index that identifies the position of the object inside the collection.  Unless otherwise stated, we assume that all the operations on these indices are taken modulo the length of the collection.  Furthermore, we may refer to negative indices and to indices greater than the length of the collection.  In these cases, indices should also be understood modulo the length of the collection.  

\subsection{Circular-arc models}

A \emph{circular-arc model} $\M$ is an ordered pair $(C(\M), \A(\M))$ where $C(\M)$ is a circle and $\A(\M)$ is a finite family of open arcs of $C(\M)$.  The arcs in $\A(\M)$ are said to be \emph{arcs} of $\M$, and $C(\M)$ is said to be the \emph{circle} of $\M$.  Unless explicitly stated, we always choose the clockwise direction for traversing $C(\M)$. For $s, t \in C(\M)$, write $(s, t)$ to mean the open arc of $C(\M)$ defined by traversing the circle from $s$ to $t$.  Call $s, t$ the \emph{extremes} of $(s, t)$, while $s$ is the \emph{beginning point} and $t$ is the \emph{ending point} of the arc.  For each $A \in \A(\M)$, represent by $s(A)$ the beginning point of $A$ and by $t(A)$ the ending point of $A$.  The \emph{extremes} of $\M$ are those of all the arcs $A \in \A(\M)$.  

Throughout this paper, we assume that no pair of extremes of a circular-arc model coincide and that no single arc entirely covers the circle of the model \citep{Golumbic2004}.  A \emph{segment} of $\M$ is an arc of $C(\M)$ formed by two consecutive extremes of $\M$.  Say that $\epsilon > 0$ is \emph{small enough} if $\epsilon < \ell$, where $\ell$ is the minimum among the lengths of all the segments of $\M$. \emph{Duplicating the arc $A$} means inserting the arc $(s(A) + \epsilon, t(A) + \epsilon)$ into $\M$, for some small enough $\epsilon$.  An \emph{$s$-sequence} is a maximal sequence of consecutive beginning points.  Similarly, a \emph{$t$-sequence} is a maximal sequence of consecutive ending points.  In general, an \emph{extreme sequence} means either an $s$-sequence or a $t$-sequence.   The \emph{complement} of $A \in \A(\M)$ is the arc $\overline{A} = (t(A), s(A))$, and the \emph{complement} of $\M$ is the circular-arc model $\overline{\M} = (C(\M), \{\overline{A} \mid A \in \A(\M)\})$.

There are five elementary classes of circular-arc models that are of special interest \citep{LinSzwarcfiterDM2009}.  Let $\M$ be a circular-arc model.  If every family of pairwise intersecting arcs share a common point, then $\M$ is a \emph{Helly} circular-arc (HCA) model.  If $\M$ has no pair of arcs that together cover the circle, then $\M$ is a \emph{normal} circular-arc (NCA) model.  When there is no arc contained in any other arc, we say that $\M$ is a \emph{proper} circular-arc (PCA) model.  If in addition all the arcs have the same length, then $\M$ is a \emph{unit} circular-arc (UCA) model.  Finally, $\M$ is an \emph{interval} circular-arc (ICA) model when some point of $C(\M)$ is covered by no arcs.

An \emph{interval model} is a finite family $\I$ of open intervals on the real line.  It is easy to see that every interval circular-arc model $\M$ is in a one-to-one correspondence with an interval model $\I$.  For this reason we will use interval models in replace of ICA models, and every definition on circular-arc models is translated to interval models with such correspondence.  

The five elementary classes of circular-arc models can be combined so as to generate a total of $32$ classes of circular-arc models, as follows.  Let $X \subseteq \{$N, P, U, H, I$\}$.  Say that $\M$ is an \emph{$X$CA} model if $\M$ is an $x$CA model, for every $x \in X$.  For instance, $\M$ is an \{N,H\}CA model if $\M$ is both an NCA and an HCA model.  Clearly, if $\M$ is an $X$CA model, then $\M$ is also an $Y$CA for every $Y \subseteq X$.  Not all the $32$ classes of circular-arc models are different, because some of the properties are implied by others.  For example, every UCA model is proper, thus every UCA  model is a \{U,P\}CA model as well.  Similarly, every interval model is Helly and normal.  This leaves us with $15$ different classes of circular-arc models that are listed in Table~\ref{tab:CA models}.  To avoid the ugly set notation to name a class of circular-arc models, we choose a better acronym for each class of circular-arc models.

\begin{table}
 \begin{center}
 \begin{tabular}{|l|l|c|c|c|c|c|}
  \hline %
  Name of the model                & Acronym  & I & U & P & H & N \\ 
  \hline %
  Circular-arc                     & CA       & 0 & 0 & 0 & 0 & 0 \\ 
  Normal circular-arc              & NCA      & 0 & 0 & 0 & 0 & 1 \\
  Helly circular-arc               & HCA      & 0 & 0 & 0 & 1 & 0 \\
  Normal Helly circular-arc        & NHCA     & 0 & 0 & 0 & 1 & 1 \\
  Proper circular-arc              & PCA      & 0 & 0 & 1 & 0 & 0 \\
  Normal proper circular-arc       & PNCA     & 0 & 0 & 1 & 0 & 1 \\
  Proper Helly circular-arc        & PHCA     & 0 & 0 & 1 & 1 & 0 \\
  Normal proper Helly circular-arc & NPHCA    & 0 & 0 & 1 & 1 & 1 \\
  Unit circular-arc                & UCA      & 0 & 1 & 1 & 0 & 0 \\
  Normal unit circular-arc         & NUCA     & 0 & 1 & 1 & 0 & 1 \\
  Unit Helly circular-arc          & UHCA     & 0 & 1 & 1 & 1 & 0 \\
  Normal unit Helly circular-arc   & NUHCA    & 0 & 1 & 1 & 1 & 1 \\
  Interval                         & IG       & 1 & 0 & 0 & 1 & 1 \\
  Proper interval                  & PIG      & 1 & 0 & 1 & 1 & 1 \\
  Unit interval                    & UIG      & 1 & 1 & 1 & 1 & 1 \\
  \hline
 \end{tabular}
 \end{center}
\caption[Subclasses of circular-arc models]{Subclasses of circular-arc models.  The five columns on the right show the properties that each class of model satisfies.}\label{tab:CA models}
\end{table}

\subsection{Circular-arc graphs}

A graph $G$ is a \emph{circular-arc graph} if there is a one-to-one correspondence between the vertices of $G$ and a family $\A$ of arcs of a circle $C$ such that two vertices are adjacent if and only if their corresponding arcs have nonempty intersection. The circular-arc model $(C, \A)$ is called a \emph{model} or \emph{representation} of $G$, and $G$ is said to \emph{admit} the model $(C, \A)$.  In other words, $G$ is a circular-arc graph if it is isomorphic to the intersection graph of $\A$.  Say that two circular-arc models are \emph{equivalent} when their intersection graphs are isomorphic.

By restricting the attention to a subclass of circular-arc models, we obtain a special class of circular-arc graphs, formed by those graphs that admit a model of the subclass.  For $X \subseteq \{$N, P, U, H, I$\}$, say that $G$ is an \emph{$X$CA} graph when $G$ admits an $X$CA model.  As before, these $32$ classes of circular-arc graphs are not all different.  Not even the fifteen classes defined by the special models in Table~\ref{tab:CA models} are all different, because graphs may admit many circular-arc models.  

\begin{theorem}[\citealp{Roberts1969}, see Theorem~\ref{thm:roberts}]
Every PIG model is equivalent to a UIG model.
\end{theorem}

\begin{theorem}[\citealp{TuckerDM1974}, see Theorem~\ref{thm:PCA->normal}]
Every PCA model is equivalent to an NPCA model.
\end{theorem}

In total, ten different subclasses of circular-arc graphs are obtained by combining the normal, proper, unit, Helly, and interval properties of a circular-arc model.  These classes are enumerated in Table~\ref{tab:CA graphs}.  As before, we avoid the ugly set notation by defining appropriate acronyms for each subclass.  For historic reasons, proper interval graphs are also called unit interval graphs and the acronym UIG is also used for PIG graphs.

\begin{table}
 \begin{center}
 \begin{tabular}{|l|l|l|}
  \hline%
  Name of the class               & Acronyms  & Admitted models \\ 
  \hline%
  Circular-arc                    & CA        & CA  \\
  Normal circular-arc             & NCA       & NCA \\
  Helly circular-arc              & HCA       & HCA \\
  Normal Helly circular-arc       & NHCA      & NHCA \\
  Proper circular-arc             & PCA       & PCA and NPCA \\
  Proper Helly circular-arc       & PHCA      & PHCA and NPHCA \\
  Unit circular-arc               & UCA       & UCA and NUCA \\
  Unit Helly circular-arc         & UHCA      & UHCA and NUHCA \\
  Interval                        & IG        & interval \\
  Proper interval (unit interval) & PIG (UIG) & UIG and PIG \\
  \hline
 \end{tabular}
 \end{center}
\caption[Subclasses of circular-arc graphs]{Subclasses of circular-arc graphs, according to the class of models that they admit.}\label{tab:CA graphs}
\end{table}

In this paper we will show characterizations, by means of forbidden induced subgraphs, for the three subclasses of NHCA graphs.  For PHCA and UHCA graphs we are able to describe the family of forbidden induced subgraphs, while for NHCA graphs we only show those circular-arc graphs that are not NHCA.  The characterizations for the three NHCA subclasses yield $O(n+m)$ time algorithms for the related recognition problems.  Such characterizations and algorithms are known for almost all the classes of circular-arc graphs listed in Table~\ref{tab:CA graphs}.  For the classes of interval, PIG, UCA, and PCA graphs the families of forbidden induced subgraphs are known (see \citealp{LekkerkerkerBolandFM1962/1963} for interval graphs, \citealp{Roberts1969} for PIG graphs, and \citealp{TuckerDM1974} for UCA and PCA graphs), while, for HCA graphs, only the list of forbidden circular-arc graphs is known \citep{JoerisLinMcConnellSpinradSzwarcfiterA2011,LinSzwarcfiter2006}.  With respect to the recognition problem, there are $O(n+m)$ time recognition algorithms for the classes of CA, PCA, UCA, HCA, interval, and PIG graphs (see \eg, \citealp{KaplanNussbaum2006a,McConnellA2003} for CA, \citealp{DengHellHuangSJC1996,KaplanNussbaum2006} for PCA, \citealp{KaplanNussbaum2006,LinSzwarcfiterSJDM2008} for UCA, \citealp{JoerisLinMcConnellSpinradSzwarcfiterA2011,LinSzwarcfiter2006} for HCA, \citealp{BoothLuekerJCSS1976} for interval graphs, and~\citealp{CorneilKimNatarajanOlariuSpragueIPL1995,DengHellHuangSJC1996} for PIG), but there is no known polynomial time recognition algorithm for NCA graphs.  There is, however, an $O(n^5m^6\log m)$ time algorithm to recognize co-bipartite NCA graphs~\citep{MullerDAM1997,HellHuangJGT2004}.

Let $\epsilon$ be a small enough value for a circular-arc model $\M$.  If we replace an arc $A \in \A(\M)$ with the arc $(s(A) + \epsilon, t(A) + \epsilon)$, we obtain a model $\M'$ equivalent to $\M$.  This implies that any circular-arc graph admits an infinite number of equivalent models.  However, $\M$ and $\M'$ are essentially the same model, since all we are interested is in how do the arcs of $\M$ intersect.  Say that two circular-arc models $\M$ and $\M'$ have \emph{equal extremes}, or simple that $\M$ and $\M'$ are \emph{equal}, if their extremes appear in the same order.  With this definition, every circular-arc graph admits a finite number of non-equal models.  When we informally say that a graph $G$ \emph{admits $k$ models}, we mean that $G$ admits $k$ non-equal models.

We will use the same terminology used for graphs and vertices when talking about circular-arc models and arcs. For example, we say that an arc is \emph{universal} to mean that its corresponding vertex is universal in the intersection graph.  Similarly, we call a model \emph{connected} when its intersection graph is connected.  For $k \in \mathbb{N}_0$, define $U_k(\M)$ to be the model obtained from $\M$ by removing all but $k$ of its universal arcs, if existing.  Clearly, $\M = U_k(\M)$ if and only if $\M$ has at most $k$ universal arcs.

In a circular-arc model, all the arcs that cover some point $p$ of the circle form a complete set. If this complete set is also a clique $\Q$, then we call $p$ a \emph{clique point} and we say that $\Q$ is \emph{represented} by $p$.  Recall that, by definition, a circular-arc model is Helly precisely when every family of pairwise intersecting arcs share a common point.  In other words, an HCA model is a circular-arc model in which every clique is represented by a clique point.

A family of arcs $\A$ of a circular-arc model $\M$ is \emph{twin-consecutive} when both the set of beginning points and the set of ending points of $\A$ correspond to consecutive sequences.  Clearly, every family of twin-consecutive arcs is formed by twin arcs, but the converse is not necessarily true.  We say that $\M$ is \emph{twin-consecutive} if every maximal family of twin arcs is twin-consecutive.  The following lemma shows that every $X$CA graph admits a twin-consecutive $X$CA model.

\begin{lemma}\label{lem:no_mellizos modelo}
 Let $\M$ be an $X$CA model for $X \subset$ \{I, U, P, H, N\}, and $A$ be an arc of $\M$.  Then, the model $\M'$ that is obtained by duplicating the arc $A$ is an $X$CA model.
\end{lemma}

\begin{corollary}\label{cor:no_mellizos grafo}
 Let $v$ and $w$ be two twin vertices of a graph $G$. Then $G$ is an $X$CA graph if and only if $G \setminus \{v\}$ is an $X$CA graph, for every $X \subset$ \{I, U, P, H, N\}.  Furthermore, $\M$ is an $X$CA model of $G \setminus \{v\}$ where the arc $A$ corresponds to $w$ if and only if the model obtained by duplicating $A$ in $\M$ is an $X$CA model of $G$.
\end{corollary}

\citet{TuckerDM1974} defined the class of $CI(n,k)$ graphs so as to characterize those PCA graphs that are not UCA.  For relative prime values $n$ and $k$ such that $n > 2k$, define $CI(n, k)$ as the circular-arc model that is built as follows (see also Figure~\ref{fig:CInk}).  Let $C$ be a circle of length $4n$.  Draw $n$ arcs $A_0, \ldots, A_{n-1}$ of length $4k + 1$ such that each $A_i$ begins at $4ki$ and ends at $4k(i+1) +1$.  Afterwards, draw $n$ arcs $B_0, \ldots, B_{n-1}$ of length $4k-1$ such that each $B_i$ begins at $4ki + 2k + 1$ and ends at $4k(i+1) + 2k$.  The intersection graph of the model $CI(n, k)$ is called the $CI(n, k)$ graph.  Note that every $CI(n,k)$ model is an NPCA model by definition.

\begin{theorem}[\citealp{TuckerDM1974}]\label{thm:CInk}
Let $G$ be a PCA graph.  Then $G$ is a UCA graph if and only if it contains no $CI(n,k)$ as an induced subgraph, with $n$ and $k$ relative primes and $n > 2k$.
\end{theorem}

\begin{figure}
 \centering
 \hspace*{\stretch{2}}\parbox[c]{3.9cm}{\includegraphics{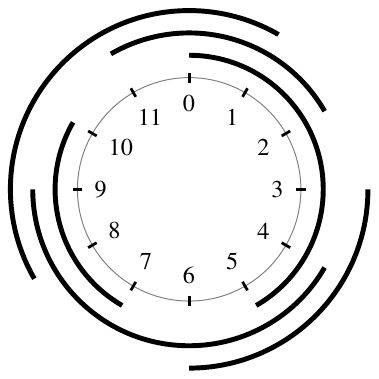}} \hspace*{\stretch{1}} \parbox[c]{2.9cm}{\includegraphics{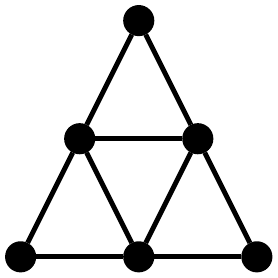}}\hspace*{\stretch{2}}
 \caption{The circular-arc model $CI(3,1)$ and its intersection graph.}\label{fig:CInk}
\end{figure}

As we already mentioned, \citet{TuckerDM1974} \citep[see also][]{Golumbic2004} proved that every PCA graph admits an NPCA model.  We now give a rather elementary proof of this fact which, we believe, is much simpler than the one in \citep{Golumbic2004}.  For this, we need to use the following simple lemma.

\begin{lemma}\label{lem:two covering arcs PCA}
 Let $\M$ be a PCA model of a graph $G$.  If $A_1, A_2$ are two arcs of $\M$ that together cover the circle, then both $A_1$ and $A_2$ are universal.
\end{lemma}

\begin{proof}
 The order in which the extremes of $A_1$ and $A_2$ appear in a traversal of $C(\M)$ is $s(A_1)$, $t(A_2)$, $s(A_2)$, $t(A_1)$ because $A_1$ and $A_2$ cover the circle.  If some arc $A$ has its ending point $t(A)$ outside $A_1$, then $t(A) \in A_2$.  If in addition $s(A) \in A_2 \setminus A_1$, then either $A \subset A_2$ or $A \supset A_1$, which contradicts the fact that $\M$ is a PCA model.  Consequently, every arc has either its ending or its beginning point inside $A_1$, \ie, $A_1$ is universal.  The proof for $A_2$ is analogous. 
\end{proof}

Tucker's theorem is now a simple corollary.

\begin{theorem}[\citealp{Golumbic2004,TuckerDM1974}]\label{thm:PCA->normal}
 Every PCA model is equivalent to an NPCA model.
\end{theorem}

\begin{proof}
 Let $\M$ be a PCA model with $k$ universal arcs.  By Lemma~\ref{lem:two covering arcs PCA}, $U_1(\M)$ is an NPCA model because it has only one universal arc.  If $k > 1$, then duplicate $k - 1$ times the universal arc of $U_1(\M)$ so as to include all the universal arcs that were possibly removed from $\M$ in the process of building $U_1(\M)$.  The resulting model is NPCA by Lemma~\ref{lem:no_mellizos modelo}.
\end{proof}

The strength of Lemma~\ref{lem:two covering arcs PCA} lies not only in its simplicity and the fact that it implies Theorem~\ref{thm:PCA->normal}.  It also helps us to transform a PCA model of a graph $G$ into an NPCA model of $G$, in $O(n)$ time (see Section~\ref{sec:algorithms}).  And, it can be used to find all the universal arcs of $\M$ as follows.

\begin{lemma}\label{lem:universal-arcs}
 An arc $A$ of a PCA model $\M$ is universal if and only if $A$ contains at least $n-1$ extremes of $\M$.
\end{lemma}

\begin{proof}
 Let $k$ be the number of extremes that appear inside $A$.  If $k < n-1$ then there is at least one arc $A'$ that has both of its extremes outside $A$.  Since $A$ is not contained in any other arc, then it follows that $A \cap A' = \emptyset$, \ie, $A$ is not universal.  If $k \geq n-1$ then either all the arcs have at least one extreme inside $A$ or there is one arc whose both extremes lie inside $A$.  In the first case $A$ is universal by definition, while in the second case $A$ is universal by Lemma~\ref{lem:two covering arcs PCA}.
\end{proof}

\section{Why study the subclasses of NHCA graphs?}
\label{sec:motivation}

One of the most important algorithmic properties about interval models is that every family of pairwise intersecting intervals share a common point.  The first linear-time recognition algorithm for interval graphs \citep{BoothLuekerJCSS1976} is based on this fact, as well as are some of the newer recognition algorithms \citep[\eg,][]{HabibMcConnellPaulViennotTCS2000}.  When interval models are generalized to circular-arc models, the Helly property is completely lost.  The class of Helly circular-arc graphs lies between the CA and IG classes and, for this reason, HCA graphs preserve a lot of nice properties of interval graphs that are lost even for UCA graphs.  Usually, these properties involve the cliques of the graphs.  Just to give one of the many examples, consider the two theorems below.

\begin{theorem}[\citealp{GilmoreHoffmanCJM1964,FulkersonGrossPJM1965}]\label{thm:IG cliques consecutivas}
Let $G$ be a graph.  Then $G$ is an interval graph if and only if there is an ordering of the cliques of $G$ such that, for every vertex $v$, the cliques containing $v$ appear form a linear range of the ordering.
\end{theorem}

\begin{theorem}[\citealp{GavrilN1974}]\label{thm:caracterizacion HCA por orden de vertices}
Let $G$ be a graph.  Then $G$ is an HCA graph if and only if there is an ordering of the cliques of $G$ such that, for every vertex $v$, the cliques containing $v$ form a range of the ordering.
\end{theorem}

It is not hard to find a counterexample to Theorem~\ref{thm:caracterizacion HCA por orden de vertices} when HCA graphs are replaced with circular-arc graphs.  Of course, there are also some properties that hold for interval graphs but do not hold for HCA graphs, because in an HCA model there could be two arcs covering the circle.  In a similar way, we can expect the jump from UIG graphs to UCA graphs to be as big as it is the jump from interval graphs to circular-arc graphs, as well as the jump from UIG graphs to PCA graphs.

In this section we motivate the study of the NHCA subclasses by showing some of the properties that are lost because of these jumps.  That is, we show several properties that hold in the NHCA subclasses and are not easy to generalize to the more general classes of circular-arc graphs.  These properties are generalizations or slight modifications of properties that hold for the interval graph subclasses, in the same sense as Theorems \ref{thm:IG cliques consecutivas}~and~\ref{thm:caracterizacion HCA por orden de vertices}.

The first property that we describe is what we call the ``local interval'' property.  Clearly, in an interval model $\I$, the submodel of $\I$ induced by the closed neighborhood of an interval $I$ is an interval model, for every $I \in \I$.  This property is not true for general HCA models and for general NCA models.  Even more, there are some graphs, such as the $4$-wheel graph (see Figure~\ref{fig:forbidden-NHCA-graphs}), for which no circular-arc model satisfies the local interval property.  On the other hand, if we restrict our attention to an NHCA model $\M$, we can see that the submodel of $\M$ induced by the closed neighborhood of an arc $A$ is an interval model, for every $A \in \A(\M)$.  Hence, NHCA graphs will retain many properties of interval graphs that deal with the ``intervality'' of the neighborhood of a vertex.  The local intervality property of NHCA models follows from the next theorem by Lin and Szwarcfiter.

\begin{theorem}[\citealp{JoerisLinMcConnellSpinradSzwarcfiterA2011,LinSzwarcfiter2006}] \label{thm:modelo-hca}
 A circular-arc model $\M$ is HCA if and only if
 \begin{description}
  \item{(i)} if three arcs of $\M$ cover $C(\M)$, then two of them also cover it, and
  \item{(ii)} the intersection graph of $\overline{\M}$ contains no induced hole.
 \end{description}
\end{theorem}

\begin{corollary} \label{cor:pca_no_hca}
 If an NCA model $\M$ is not HCA, then three arcs of $\M$ cover $C(\M)$.
\end{corollary}

\begin{proof}
Suppose that no three arcs of $\M$ cover $C(\M)$ and yet $\M$ is not HCA.  Then, by Theorem \ref{thm:modelo-hca}, the intersection graph of $\overline{\M}$ has an induced hole $v_1, \ldots, v_k$, for some $k \geq 4$.  Thus, the arcs $\overline{A_1}$, $\overline{A_3}$ $\in \A(\overline{\M})$ corresponding to vertices $v_1$ and $v_3$ do not intersect and, therefore, $A_1$ and $A_3$ are arcs of $\A(\M)$ that cover $C(\M)$, \ie, $\M$ is not NCA.
\end{proof}

As a corollary we obtain Theorem~\ref{thm:NHCA->no2y3cubrenCirculo} below.  For the sake of simplicity, we will take this property as an alternative definition of NHCA, PHCA, and UHCA graphs.  Hence, sometimes we omit the reference to this theorem.

\begin{theorem}\label{thm:NHCA->no2y3cubrenCirculo}
 A graph is an $X$HCA graph if and only if it admits an $X$CA model in which there are no two nor three arcs covering the circle, for $X \in \{N, P, U\}$.
\end{theorem}

\begin{corollary}[Local interval property for models]
 A circular-arc model $\M$ is an NHCA model if and only if, for every $A \in \A(\M)$, the submodel of $\M$ induced by $N[A]$ is an interval model.
\end{corollary}

The local interval property of NHCA graphs implies that a lot of simple algorithms that work for interval graphs also work for NHCA graphs.  Consider for example the problem of finding every clique of an interval graph.  A \emph{clique segment} in an interval model $\M$ is a segment $(s, t)$ where $s$ is a beginning point and $t$ is an ending point.  In $\M$, every clique point belongs to a clique segment, and all the points in a clique segment are clique points.  Thus, the set of intervals that contain a given clique segment induce a clique and each clique is represented by exactly one clique segment.  It is trivial to find every clique of an interval graph by computing the clique segments in one of its interval models.  For HCA graphs the situation is quite different.  First, not every segment of the form $(s, t)$ is formed by clique points (with $s$ and $t$ being two consecutive beginning and ending point, respectively).  Second, not every clique of the model is represented by exactly one clique segment, \ie, there are several clique segments whose containing arcs induce the same clique.  Thus, the $O(n)$ time algorithm to find every clique point is not so trivial \citep[see][]{LinMcConnellSoulignacSzwarcfiter2008,Soulignac2010}.  However, NHCA graphs share the same property as interval graphs: every segment $(s, t)$ represents a clique, and every clique is represented by exactly one of such clique segments.  Thus, the same algorithm for interval graphs works as well for NHCA graphs.  Another example in where the local interval property has some advantages is in the dynamic recognition problem.  It is somehow easy to insert a new edge in $O(1)$ time into a PHCA graph, but it is not so easy for PCA graphs \citep[see][]{Soulignac2010}.  In Section~\ref{sec:orientations} we further develop these intervality concepts for NHCA and PHCA graphs.

Another similarity between interval and NHCA graphs has to do with the forbidden induced subgraph structure.  For interval graphs the structure is described by the following theorem by Roberts.  Here $K_{1,3}$ is the graph with four vertices, where one of these vertices is universal and the other three vertices are pairwise non-adjacent.

\begin{theorem}[\citealp{Roberts1969}]\label{thm:roberts}
Let $G$ be an interval graph.  Then the following are equivalent:
\begin{enumerate}[(i)]
 \item $G$ does not contain $K_{1,3}$ as an induced subgraph.
 \item $G$ is a proper interval graph.
 \item $G$ is a unit interval graph.
\end{enumerate}
\end{theorem}

For NHCA graphs we obtain a similar result.

\begin{theorem}\label{thm:caracterizacion-NHCA->PHCA}
Let $G$ be an NHCA graph.  Then $G$ is a PHCA graph if and only if $G$ contains no $K_{1,3}$ as an induced subgraph.
\end{theorem}

\begin{proof}
 Clearly, the $K_{1,3}$ graph is not PCA, so it is neither PHCA.  For the converse, suppose that $G$ is an NHCA graph with an NHCA model $\M$.  Sort every extreme sequence of $\M$ in such a way that no arc with an extreme in the sequence is properly contained in some other arc with an extreme in the sequence.  This sorting does not change the intersections between the arcs of $\M$, thus the sorted model $\M'$ is also an NHCA model of $G$.  Now, if some arc $A_1 \in \A(\M')$ is contained in some other arc $A_2 \in \A(\M')$, it is because there is some ending point between $s(A_1)$ and $s(A_2)$, and there is some beginning point between $t(A_2)$ and $t(A_1)$.  In other words, there are two arcs $L$ and $R$ of $\M'$ such that $s(A_1), t(L), s(A_2), t(A_2), s(R)$ and $t(A_1)$ appear in this order in $\M'$.  Since $\M'$ is NHCA then $L \neq R$ and $L \cap R = \emptyset$, \ie, the intersection graph of $A_1, L, A_2, R$ is isomorphic to $K_{1,3}$.
\end{proof}

Implication $(ii) \Longrightarrow (iii)$ of Theorem~\ref{thm:roberts} is lost, because the graph $CI(n,k)$ is PHCA but not UHCA for every $n > 3k$.  Indeed, by definition, every $CI(n,k)$ graph admits a PCA model in which no family with at most $\lceil n/k \rceil > 3$ arcs cover the circle.  To retain the ``proper=unit'' property with a similar definition as the one in Theorem~\ref{thm:NHCA->no2y3cubrenCirculo}, we should ask that no set of arcs cover the circle.  This is because the unit length property of UIG graphs is global, and it does not depend only on the neighborhood of each vertex.

Next, we consider clique graphs of HCA graphs and NHCA graphs.  The \emph{clique graph} $K(G)$ of a graph $G$ is the intersection graph of the cliques of $G$, and for a class $\C$ of graphs we denote $K(\C) = \{K(G) \mid G \in \C\}$.  Clique graphs of interval graphs were studied by \citet{HedmanJCTSB1984}, who proved that the PIG and $K$(IG) classes are equal.  On the other hand, clique graphs of HCA graphs were first studied by \citet{DuranLinAC2001}, who proved that clique graphs of HCA graphs are both PCA and HCA.  The question that motivated us to study PHCA graphs first, and NHCA graphs later, was if the class of PHCA graphs is equal to the $K$(HCA) class.  The answer is no, but almost, as it is shown in the next theorem.

\begin{theorem}[\citealp{LinSoulignacSzwarcfiterDAM2010}]{\label{thm:KHCA}}
Let $H$ be a graph and $U$ the set of universal vertices of $H$.  Then $H$ is the clique graph of some HCA graph $G$ if and only if:
\begin{description}
\item{$\;$(i)} $H$ is a PHCA graph or
\item{(ii)} $H \setminus U$ is a co-bipartite PHCA graph and $|U| \geq 2$.
\end{description}
\end{theorem}

However, the analogous of Hedman's result can be obtained for NHCA graphs \citep{LinSoulignacSzwarcfiterDAM2010}.  That is, the $K$(NHCA) and the PHCA classes of graphs are equal.  \citet{HedmanJCTSB1984} proved also that for every PIG graph $G$ there is a PIG graph $H$ such that $K(H)$ is isomorphic to $G$.  That is, the PIG and the $K$(PIG) classes are also equal.  The same result holds for PHCA graphs, \ie, for every PHCA graph $G$ there is a PHCA graph $H$ such that $K(H)$ and $G$ are isomorphic \citep{LinSoulignacSzwarcfiterDAM2010}.

Now we move into the forth property that can be preserved by restricting the attention to PHCA graphs.  A classic characterization of interval graphs is that interval graphs are those graphs whose clique matrix has the consecutive-ones property for columns \citep{GilmoreHoffmanCJM1964}.  Similarly, a graph is HCA if and only if its clique matrix has the circular-ones property for columns \citep{GavrilN1974}.  The definitions of clique matrix, consecutive and circular-ones properties are given in Section~\ref{sec:properties}.  We can think that the consecutive and circular-ones properties for columns are due to the Helly property of the interval and HCA graphs, respectively.  On the other hand, it is well known that a graph is a proper interval graph if and only if its clique matrix has the consecutive-ones property for both its rows and its columns \citep[see \eg][]{DeogunGopalakrishnanIPL1999,Fishburn1985,GardiDM2007}.  The analogous theorem for the circular-ones property can be proved for PHCA graphs, as we shall see in Section~\ref{sec:properties}.  That is, the clique matrix of a graph has the circular-ones properties for both rows and columns if and only if the graph is a PHCA graph.

Finally, the boxicity of NHCA graphs was studied by \citet{BhowmickChandranGaC2010}.  Given two graphs $G_1$ and $G_2$ with the same vertex set $V$, define the \emph{intersection} of $G_1$ and $G_2$ as the graph $G_1 \cap G_2$ with vertex set $V$ such that $vw$ is an edge of $G_1 \cap G_2$ if and only if $vw$ is an edge of both $G_1$ and $G_2$. The \emph{boxicity} of a graph $G$ is the minimum number of interval graphs whose intersection is isomorphic to $G$.  Clearly, interval graphs are precisely the graphs with boxicity 1.  On the other hand, the boxicity of a general circular-arc graph can be as large as $n/2$ \citep{Roberts1969a}.  \citet{BhowmickChandranGaC2010} proved that NHCA graphs have boxicity at most $3$.  If, furthermore, no four arcs of a circular-arc model cover the circle, then the boxicity is at most $2$.

\section{The structure of the NHCA subclasses}
\label{sec:characterizations}

In this section we present characterizations by forbidden induced subgraphs for the classes of NHCA, PHCA, and UHCA graphs.  These characterizations follow the same spirit as Theorem~\ref{thm:caracterizacion-NHCA->PHCA}, in the sense that they show when does a graph from some class belong to a subclass of it.  The characterizations shown in this section immediately yield $O(n+m)$ time recognition algorithms for all the classes.  In the next section we further discuss the algorithmic implications of these characterizations.

The following proposition is used several times throughout this section.  We include it here for this reason.

\begin{proposition}\label{prop:holes in ca models}
Let $\M$ be a circular-arc model and $B_1, \ldots, B_k$ be a hole in $\M$.  If $A \in \A(\M)$ is an arc which is not contained in any other arc, then either:

\begin{enumerate}[(i)]
 \item $A$ and $B_i$ cover the circle for some $1 \leq i \leq k$,
 \item $A, B_i$, and $B_{i+1}$ cover the circle for some $1 \leq i \leq k$, 
 \item $A \subset (B_i \cup B_{i+1}) \setminus (B_{i+2} \cup \ldots \cup B_{i-1})$ for some $1 \leq i \leq k$, or
 \item $A, B_i, \ldots, B_j$ is an induced hole of $\M$, for some $1 \leq i,j \leq k$.\label{prop:holes in ca models:condition iv}
\end{enumerate}
\end{proposition}

\begin{proof}
If $A = B_i$ for some $1 \leq i \leq n$, then ($\ref{prop:holes in ca models:condition iv}$) follows.  Suppose then that $A$ is not an arc of the hole.  Traverse $C(\M)$ from $t(A)$ and let $B_i$ be the arc whose beginning point appears first.  If $s(B_i) \in A$ then $A$ and $B_{i-1}$ must cover the circle.  Otherwise, if $t(B_i) \in A$ then $A, B_{i-1},$ and $B_{i}$ cover the circle.  Finally, suppose that neither $s(B_i)$ nor $t(B_i)$ are points of $A$, and let $B_j$ be the arc whose ending point appears first in a counterclockwise traversal of $C(\M)$ from $s(A)$.  If $i-1 = j+1$, then it follows that $A \subset B_{i-1}$, which is a contradiction to the fact that $A$ is not properly contained in any other arc.  Otherwise, $A, B_{i-1}, \ldots, B_{j+1}$ induce a hole whenever $i-1 \neq j+2$ or $A \subset (B_{i-2} \cup B_{i-1})  \setminus (B_{i} \cup \ldots \cup B_{i-3})$ whenever $i-1 = j+2$.
\end{proof}

\subsection{Normal Helly circular-arc graphs}
\label{sec:characterizations:nhca}

We begin with the problem of determining when does an HCA graph admit an NHCA model.  Wheels, $3$-suns, rising suns, and umbrellas are the forbidden subgraphs involved in the characterization.  The $3$-sun and the umbrella are the graphs depicted in Figure~\ref{fig:forbidden-NHCA-graphs} (a)~and~(b), respectively.  The \emph{$n$-wheel}, for $n \geq 4$, is the graph obtained by inserting one universal vertex into a hole of length $n$ (see Figure~\ref{fig:forbidden-NHCA-graphs} (c)).  Finally, the \emph{$n$-rising sun}, for $n \geq 4$, is the graph that is obtained from a path $v_2, \ldots, v_{n-1}$ by first adding two universal vertices $v_1$ and $v_n$, and then inserting three vertices $w_1, w_{n-1}, w_n$ such that $w_i$ is adjacent only to $v_i$ and to $v_{i+1}$, for $i \in \{1, n-1, n\}$ (see Figure~\ref{fig:forbidden-NHCA-graphs} (d)).  The $3$-sun graph is denoted by $S_3$, the umbrella is denoted by $U$, the $n$-wheel is denoted by $W_n$, and the $n$-rising sun is denoted by $R_n$.  

\begin{figure}[ht!]
 \centering
 \begin{tabular*}{\textwidth}{@{\extracolsep{\fill}}*4c}
  \includegraphics{grafo-3sun} & \includegraphics{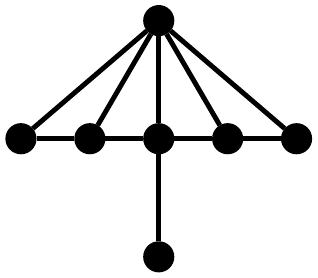} & \includegraphics{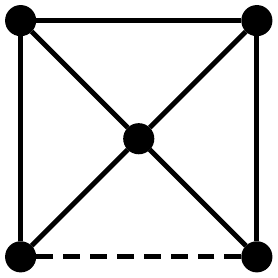} & \includegraphics{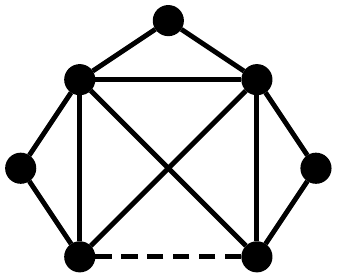} \\
  (a) $3$-sun & (b) Umbrella & (c) Wheels & (d) Rising suns 
 \end{tabular*}
 \caption[Minimal HCA graphs that are not NHCA graphs]{HCA graphs that are not NHCA.  Wheels have at least 5 vertices and rising suns have at least 7 vertices.}\label{fig:forbidden-NHCA-graphs}
\end{figure}

\begin{figure}[ht!]
 \centering
 \begin{tabular*}{\textwidth}{@{\extracolsep{\fill}}*4c}
  \includegraphics{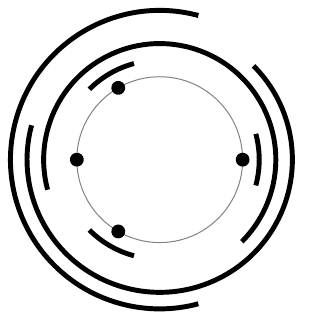} & \includegraphics{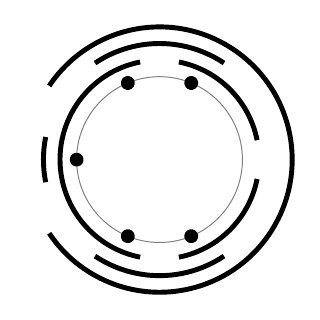} & \includegraphics{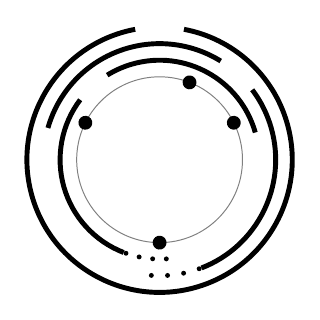} & \includegraphics{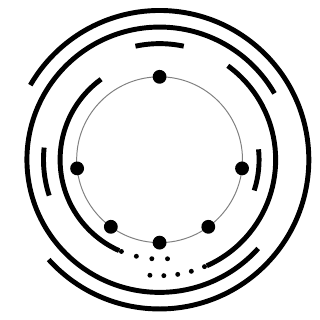} \\
  (a) $3$-sun & (b) Umbrella & (c) Wheels & (d) Rising suns 
 \end{tabular*}
 \caption[HCA models of the graphs in Figure~\ref{fig:forbidden-NHCA-graphs}]{HCA models of the graphs in Figure~\ref{fig:forbidden-NHCA-graphs}.  Points in the circles are used to mark the clique points.}\label{fig:forbidden-NHCA-helly-models}
\end{figure}

Interval graphs can be partially characterized as HCA graphs using the following result by Lekkerkerker and Boland.

\begin{theorem}[\citealp{LekkerkerkerBolandFM1962/1963}]\label{thm:lekkerkerkerboland}
 An HCA graph is an interval graph if and only if it does not contain holes, $3$-suns, rising suns, nor umbrellas as induced subgraphs.
\end{theorem}

The only difference between the characterization by Lekkerkerker and Boland and the characterization of NHCA graphs is that holes are replaced by wheels.  We analyze first how do HCA models of NHCA graphs look like.  

\begin{theorem}\label{thm:caracterizacion-HCA->NHCA}
 Let $\M$ be an HCA model of a graph $G$.  Then, $\M$ is equivalent to an NHCA model if and only if $\M$ contains no wheels, $3$-suns, rising suns, nor umbrellas as induced submodels.
\end{theorem}

\begin{proof}
 Let $\M'$ be any HCA model of a wheel $W_k$, for $k \geq 4$.  Such a model exists as it is depicted in Figure~\ref{fig:forbidden-NHCA-helly-models} (c).  Model $\M'$ has at least one clique point for each clique of $W_k$; the universal arc covers all of these clique points, while each of the other arcs of the submodel covers exactly two of them. Consequently, there are two arcs covering $C(\M')$, \ie, $\M'$ is not normal. Thus, whenever $\M$ is equivalent to an NHCA model, $\M$ does not contain an induced submodel of $W_k$. The proofs for the $3$-sun, the umbrella and the rising suns follow analogously.

 For the converse, suppose that $\M$ contains none of the forbidden submodels and yet $\M$ has two arcs $A_1$ and $A_2$ that cover the circle which, w.l.o.g., are not contained in other arcs of $\M$.  Suppose also, to obtain a contradiction, that $\M$ has $k$ arcs $B_1, \ldots, B_k$ that induce a hole in that order.  If $A_1 = B_i$ for some $1 \leq i \leq k$, then $A_2$ is adjacent to all the arcs of the hole, \ie, $A_2, B_1, \ldots, B_k$ induce a wheel.  Otherwise, we ought to consider three cases by Proposition~\ref{prop:holes in ca models}:
 
 \begin{description}
  \item[Case 1:] $A_1$ and $B_i$ cover the circle for some $1 \leq i \leq k$.  In this case, $A_1$ intersects all the arcs of the hole, \ie, $A_1, B_1, \ldots, B_k$ induce a wheel.  Therefore, this case cannot happen.
  \item[Case 2:] $A_1$ is contained in $B_i \cup B_{i+1}$ for some $1 \leq i \leq k$.  In this case, $A_2$ intersects all the arcs of $B_1, \ldots, B_k$, thus $A_2, B_1, \ldots, B_k$ induce a wheel.  This case is also impossible.
  \item[Case 3:] $A_1, B_i, \ldots, B_j$ is a hole of $\M$ for some $1 \leq i,j \leq k$.  As in Case 2, $A_2$ intersects all the arcs of this new hole.  Thus, $A_2, A_1, B_i, \ldots, B_j$ induce a wheel, again a contradiction.
 \end{description}
 
 Since none of the cases can occur, it follows that $\M$ has no induced holes.  Thus, $G$ is hole-free and it contains no rising suns, $3$-suns, nor umbrellas, which implies that $G$ is an interval graph by Theorem~\ref{thm:lekkerkerkerboland}.  Consequently, $\M$ is equivalent to some interval model of $G$.
\end{proof}

The characterization by minimal forbidden induced subgraphs then follows easily.

\begin{corollary}\label{cor:caracterizacion-HCA->NHCA}
 An HCA graph is NHCA if and only if it contains no wheels, $3$-suns, rising suns, nor umbrellas as induced subgraphs.
\end{corollary}

There is also a strong consequence for circular-arc models that can be used for negative certification and which is also useful for the characterization of NHCA graphs in terms of NCA graphs.

\begin{corollary}\label{cor:CA-models-of-NHCA}
 Every circular-arc model of a non-interval NHCA graph is NHCA.
\end{corollary}

\begin{proof}
 Let $\M$ be any circular-arc model of a non-interval NHCA graph $G$.  Since $G$ is NHCA then we can apply the algorithm by \citet{LinSzwarcfiter2006} with input $\M$, to obtain an HCA model $\M'$ of $G$.  By definition, $\M'$ is equivalent to some NHCA model, because $G$ is NHCA.  If $\M'$ has two arcs that cover the circle, then we can use the same arguments as in Theorem~\ref{thm:caracterizacion-HCA->NHCA} to prove that $G$ is an interval graph, a contradiction.  Otherwise, $\M'$ is NHCA, \ie, there are not two nor three arcs that cover $C(\M')$.  

 The algorithm by \citet{LinSzwarcfiter2006} works in such a way that every arc of $\M$ is included in some arc of $\M'$.  Thus, in $\M$ there could not be two nor three arcs covering the circle, \ie, $\M$ is NHCA.
\end{proof}

Now we proceed with the characterization of NHCA graphs in terms of NCA graphs.  In this case, the forbidden induced subgraphs are the wheels, the $3$-sun, the rising suns, the umbrella, and the \emph{tent graph} $\overline{S_3}$ (see Figure~\ref{fig:forbidden-NHCA-normal-models}).  The proof follows the same scheme as before, we analyze the NCA models of NHCA graphs and the characterization is obtained as a corollary.

\begin{figure}
 \centering
 \hspace*{\stretch{2}} \includegraphics{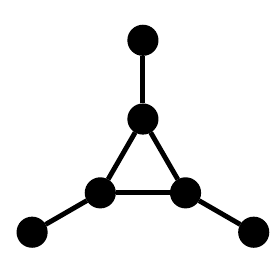}\hspace*{\stretch{1}} \includegraphics{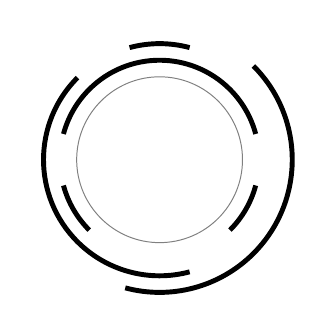} \hspace*{\stretch{2}}
 \caption{The tent graph and one of its NCA models}\label{fig:forbidden-NHCA-normal-models}
\end{figure}

\begin{theorem}[\citealp{LekkerkerkerBolandFM1962/1963}]\label{thm:lekkerkerkerboland2}
 An NCA graph is an interval graph if and only if it does not contain holes, $3$-suns, rising suns, umbrellas, nor tents as induced subgraphs.
\end{theorem}

\begin{theorem}\label{thm:caracterizacion-NCA->NHCA}
 Let $\M$ be an NCA model of a graph $G$.  Then, $\M$ is equivalent to an NHCA model if and only if $\M$ contains no wheels, $3$-suns, rising suns, umbrellas, nor tents, as induced submodels.
\end{theorem}

\begin{proof}
 Wheels, $3$-suns, rising suns, umbrellas, and tents admit circular-arc models that are not NHCA (see Figures \ref{fig:forbidden-NHCA-helly-models}~and~\ref{fig:forbidden-NHCA-normal-models}).  Hence they are not NHCA, by Corollary~\ref{cor:CA-models-of-NHCA}.

 The converse is somehow similar to the converse of Theorem~\ref{thm:caracterizacion-HCA->NHCA}, but it needs a few tweaks. Suppose that $\M$ contains none of the forbidden submodels and, yet, $\M$ has three arcs $A_1, A_2$ and $A_3$ that cover the circle.  We may assume that none of these three arcs is contained in any other arc because $\M$ has no two arcs that cover the circle.  To obtain a contradiction, suppose that $\M$ has $k$ arcs $B_1, \ldots, B_k$ that induce a hole $\MH$ in this order.
 
 \begin{description}
  \item[Claim 1:] $A_j, B_i, B_{i+1}$ do not cover the circle, for $1 \leq i \leq k$ and $1 \leq j \leq 3$.  Otherwise, $A_j$ intersects all the arcs of $\MH$, \ie, $\MH \cup \{A_j\}$ induces a wheel.
  
  \item[Claim 2:] There is a hole $\MH' \subset \MH \cup \{A_1, A_2\}$ that contains at least one of $A_1$ and $A_2$.  Suppose, to obtain a contradiction, that this is not the case.  Then, by Proposition~\ref{prop:holes in ca models} and Claim 1, it follows that $A_1 \subset (B_{i} \cup B_{i+1}) \setminus (B_{i+2}, \ldots, B_{i-1})$ and that $A_2 \subset (B_{j} \cup B_{j+1}) \setminus (B_{j+2}, \ldots, B_{j-1})$, for $1 \leq i, j \leq n$.  By Claim~1, $A_3$ does not cover the circle with $B_{i}$ and $B_{i+1}$, thus $i \neq j$.  So, either $t(A_1) \in B_{i+1} \cap A_2$ and $j = i+1$ or $t(A_2) \in B_{i} \cap A_1$ and $i = j+1$.  Assume the former w.l.o.g., thus $A_2 \subset (B_{i+1} \cup B_{i+2}) \setminus (B_{i+2}, \ldots, B_{i})$ (see Figure~\ref{fig:thm:caracterizacion-NCA->NHCA}).  Now, consider the position of arc $A_3$.  Since $A_1, A_2, A_3$ cover the circle, then $A_3$ crosses both $t(B_{i+2})$ and $s(B_{i})$.  By Claim 1, $A_3$ crosses neither $s(B_{i+1})$ nor $t(B_{i+1})$, thus $A_1, B_{i+1}, A_2, A_3, B_i$ induce a wheel where $A_1$ is the universal arc.
 \end{description} 
 
 \begin{figure}
  \centering
  \includegraphics{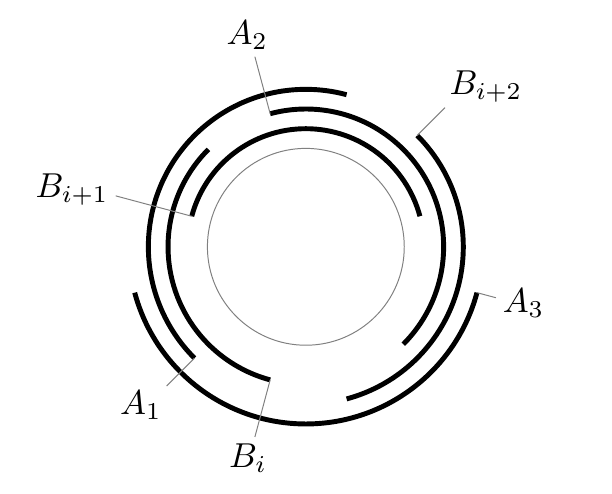}
  \caption[Claim 2 of Theorem~\ref{thm:caracterizacion-NCA->NHCA}]{Claim 2 of Theorem~\ref{thm:caracterizacion-NCA->NHCA}, $A_1, B_{i+1}, A_2, A_3, B_i$ induce a wheel.}\label{fig:thm:caracterizacion-NCA->NHCA}
 \end{figure}
 
 By Claim 2, there is a hole $\MH' \subset \MH \cup \{A_1, A_2\}$ that contains at least one of $A_1, A_2$, say $A_1$.  If we exchange $\MH$ with $\MH'$ and $A_1$ with $A_3$ in Claim 2, we obtain that there is a hole $\MH'' \subset \MH' \cup \{A_2, A_3\}$ that, w.l.o.g., contains $A_2$ as well.  If $A_1$ is not an arc of $\MH''$, it is because $A_1$ is covered by two arcs of the hole, one of which is $A_2$.  In this case, call $B$ to the arc of $\MH''$ that together with $A_2$ covers $A_1$.  Otherwise, call $B = A_1$.  Summing up, $\MH'' = B, A_2, B_{i}, \ldots, B_j$, for some pair $1 \leq i,j \leq k$.  Since $A_1, A_2, A_3$ cover the circle, then $B, A_2, A_3$ cover the circle, thus $A_3$ must intersect all the arcs of $\MH''$, a contradiction.  This contradiction appears because we assume that there is a hole in $\M$.  Therefore, $G$ contains no holes, $3$-suns, umbrellas, rising suns, nor tents as induced subgraphs.  This implies that $G$ admits an interval model equivalent to $\M$ by Theorem~\ref{thm:lekkerkerkerboland2}.
\end{proof}

\begin{corollary}\label{cor:caracterizacion-NCA->NHCA}
 An NCA graph is NHCA if and only if it contains no wheels, $3$-suns, rising suns, umbrellas, nor tents as induced subgraphs.
\end{corollary}

The proofs of Theorems \ref{thm:caracterizacion-HCA->NHCA}~and~\ref{thm:caracterizacion-NCA->NHCA} can be combined so as to obtain a characterization of when does a circular-arc graph admit an NHCA model.  We are not going to prove the theorem to avoid repetitions, instead we just give a sketch of the proof.

\begin{theorem}\label{thm:caracterizacion-CA->NHCA}
 A circular-arc graph is NHCA if and only if it contains no wheels, $3$-suns, rising suns, umbrellas, nor tents as induced subgraphs.
\end{theorem}

\begin{proof}
 Wheels, $3$-suns, rising suns, umbrellas and tents are not NHCA by Theorem~\ref{thm:caracterizacion-NCA->NHCA}.
 
 For the converse, suppose that there are two arcs $A_1, A_2$ that cover a circular-arc model $\M$, and that $\M$ has an induced hole $B_1, \ldots, B_k$.  Then, by Proposition~\ref{prop:holes in ca models}, we can either find a wheel in $\M$ as in Theorem~\ref{thm:caracterizacion-HCA->NHCA} or $A_1, B_i, B_{i+1}$ cover the circle for some $1 \leq i \leq k$.  In this last case, $A_1$ is universal to all the arcs of the hole which also implies that $\M$ contains an induced wheel.  Then, as before, $\M$ is an interval model or $\M$ is NCA.  If $\M$ is NCA, then the result follows from Theorem~\ref{thm:caracterizacion-NCA->NHCA}. 
\end{proof}

\subsection{Proper Helly circular-arc graphs}
\label{sec:characterizations:phca}

Up to this point we have characterized which circular-arc (resp.\ HCA, NCA) graphs admit an NHCA model and which NHCA graphs admit a PHCA model.  In this section we characterize which PCA graphs are also PHCA.  For the proof we may use the same arguments of Theorem~\ref{thm:caracterizacion-HCA->NHCA} to show that every PCA model of a non-interval PHCA graph is in fact a PHCA model.  This would yield an elegant short proof.  But instead, we prefer to do a constructive proof that shows how can a PCA model of an interval graph be transformed into a proper interval model.  As we will see in Section~\ref{sec:algorithms:PHCA}, this proof yields an $O(n)$ time algorithm to transform a PCA model into a PHCA model.   We remark that this proof is the same that appeared in~\citep{LinSoulignacSzwarcfiter2008b}, with some minor corrections.

\begin{theorem} \label{thm:caracterizacion-PCA->PHCA}
Let $\M$ be a PCA model of a graph $G$. Then the following are equivalent:
\begin{description}
\item{(i)} $\M$ is equivalent to a PHCA model.
\item{(ii)} $\M$ contains no induced submodel of\/ $W_4$ and $S_3$.
\item{(iii)} $U_1(\M)$ is HCA or $U_0(\M)$ is a PIG model.
\end{description}
\end{theorem}
\begin{proof}\mbox{}

$(i) \Longrightarrow (ii)$: neither $W_4$ nor $S_3$ are NHCA graphs by Theorem~\ref{thm:caracterizacion-HCA->NHCA}, thus $\M$ cannot have induced submodels of them.

$(ii) \Longrightarrow (iii)$: let $\M$ be a PCA model, containing no induced submodels of $W_4$ and $S_3$. By Lemma~\ref{lem:two covering arcs PCA}, $\M_1 = U_1(\M)$ is an NPCA model.  If $\M_1$ is not an HCA model, Corollary~\ref{cor:pca_no_hca} implies that $\M_1$ contains three arcs $A_1, A_2, A_3$ covering $C(\M_1)$.  No two arcs cover $C(\M_1)$, thus we may assume that $s(A_1), t(A_3), s(A_2), t(A_1), s(A_3), t(A_2)$ appear in this order in a traversal of $C(\M_1)$.  First, we prove that one of the above three arcs must be universal.  Suppose the contrary.  Then, there exist arcs $B_1, B_2$ and $B_3$ such that $B_i$ does not intersect $A_i$, for $i \in \{1,2,3\}$.  However, since $\M_1$ is a proper model, it follows that $B_i$ intersects $A_j, A_k$ for $\{j, k\} = \{1,2,3\} \setminus \{i\}$. The latter leads to a contradiction because the intersection graph of $\{A_i, B_i\}_{i \in \{1,2,3\}}$ is isomorphic to $S_3$, when $B_1, B_2, B_3$ are pairwise disjoint, or it contains an induced $W_4$. Consequently, one of $A_1, A_2, A_3$, say $A_1$, is the unique universal arc of $\M_1$.

Next, we examine the arc $A_1$ in $\M_1$.  We will prove that there is no pair of arcs $L, R$ such that $t(L) \in A_1$, $s(R) \in A_1$, and $s(R)$ precedes $t(L)$ in a traversal of $C(\M)$ from $s(A_1)$. To obtain a contradiction for this fact, assume the contrary and discuss the following alternatives.

\begin{description}
 \item [Case 1:] $L = A_3$.  In this situation, $R, A_2, A_3$ are three arcs covering $C$. Because $A_1$ is the unique universal arc of $\M_1$, we know that $R, A_2, A_3$ are not universal. Consequently, as above, $\M_1$ contains a submodel of $W_4$ or $S_3$, a contradiction (Figure~\ref{fig:thm_2}(a)).

 \item [Case 2:] $R = A_2$.  Similar to Case~1.

 \item [Case 3:] $L \neq A_3$ and $R \neq A_2$.  By Cases 1~and~2, above, it suffices to examine the situation where $s(R), t(L) \in (t(A_3), s(A_2))$. Suppose $L \cap A_2 = R \cap A_3 = \emptyset$. In this case, the arcs $A_1, A_2, A_3, L, R$ form a forbidden $W_4$, impossible (Figure~\ref{fig:thm_2}(b)). Alternatively, let $L \cap A_2 \neq \emptyset$. Then the arcs $A_2, L, R$ cover the circle and none of them is the universal arc $A_1$, an impossibility (Figure~\ref{fig:thm_2}(c)). The situation $R \cap A_3 \neq \emptyset$ is similar.
\end{description}

\begin{figure}
 \centering
 \begin{tabular*}{\textwidth}{@{\extracolsep{\fill}}*3c}
  \includegraphics{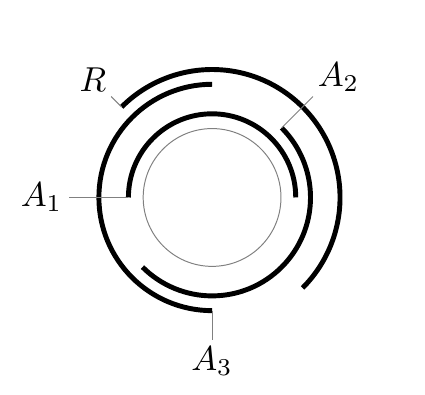} & \includegraphics{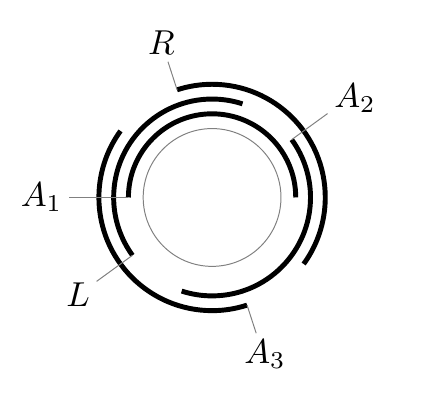} & \includegraphics{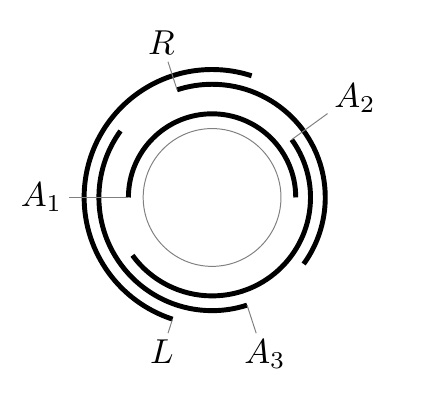} \\
  (a) & (b) & (c)
 \end{tabular*}
 \caption[Proof of Theorem~\ref{thm:caracterizacion-PCA->PHCA}]{Proof of Theorem~\ref{thm:caracterizacion-PCA->PHCA}.  If $A_1$ is completely covered, then $G$ is not a PHCA graph.}
 \label{fig:thm_2}
\end{figure}

By the above cases, we conclude that all the ending points must precede the beginning points in $(s(A_1), t(A_1))$. Let $t$ and $s$ be the last ending point and the first beginning point inside $(s(A_1), t(A_1))$, respectively. Taking into account that $A_1$ is the only universal arc, we conclude that no point of the segment $(t, s) \subset A_1$ of $C(\M)$ can be contained in any arc of $\M$ except $A_1$. Hence $\M_0 = \M_1 \setminus \{A_1\}$ is a PIG model.

$(iii) \Longrightarrow (i)$: suppose first that $\M_1 = U_1(\M)$ is an HCA model. By Lemma~\ref{lem:no_mellizos modelo}, we can include in the model all the universal arcs that have been possibly removed from it, obtaining a model equivalent to $\M$ that is both PCA and HCA.

Next, suppose that $\M_0 = U_0(\M)$ is a PIG model and $\M_1$ is not an HCA model, thus $\M_0 \neq \M_1$.  Since $\M_0 \neq \M_1$, it follows that $\M_1$ has some universal arc $A$.  We prove that $\M' := \M_0 \cup \{\overline{A}\}$ is a PIG model equivalent to $\M_1$.  By Lemma~\ref{lem:two covering arcs PCA}, $\M_1$ is an NPCA model, thus there is exactly one extreme of each non-universal arc inside $A$.  Hence, $\overline{A}$ contains exactly one extreme of each non-universal arc.  This means that $\overline{A}$ is a universal arc of $\M'$, and that $\M'$ is an NPCA model equivalent to $\M_1$.  On the other hand, since $\M_0$ is an interval model, it follows that there is some point $p \in A$ which is crossed only by $A$ in $\M_1$.  Therefore, $p$ is not crossed by any arc of $\M'$, which implies that $\M'$ is an interval model as well.  Summing up, $\M'$ is a PIG model equivalent to $\M_1$.  Duplicating the universal arc of $\M'$ we can include all the universal arcs that were removed from $\M_1$, to obtain a PIG model of $G$.  
\end{proof}

The characterization in terms of the forbidden subgraphs is depicted in the corollary below.

\begin{corollary}\label{cor:caracterizacion-PCA->PHCA}
 A PCA graph is PHCA if and only if it contains no induced\/ $W_4$ and no induced $S_3$.
\end{corollary}

In implication $(ii) \Longrightarrow (iii)$, the non-Helly PCA model $\M_1 = U_1(\M)$ has a universal arc $A$ whose removal yields the interval model $U_0(\M)$.  This arc $A$ can be replaced by $\overline{A}$, and then all the arcs of $\M \setminus \M_1$ can be inserted once again into $\M_0$ by duplicating $\overline{A}$ as in $(iii) \Longrightarrow (i)$.  The model so obtained is a PIG model equivalent to $\M$.  Thus, if some PHCA graph admits a non-Helly PCA model, then the graph is in fact a PIG graph.  The following corollary, which is the analogous of Corollary~\ref{cor:CA-models-of-NHCA}, reflects this fact.

\begin{corollary}\label{cor:PCA-models-of-PHCA}
 Every PCA model of a non-interval PHCA graph is PHCA.
\end{corollary}

\begin{proof}
 By Corollary~\ref{cor:CA-models-of-NHCA}, every PCA model of a non-interval PHCA graph is also NHCA.  Thus, the model is both PCA and HCA.
\end{proof}

\subsection{Unit Helly circular-arc graphs}

Theorem~\ref{thm:caracterizacion-PCA->PHCA} describes when can a PCA model be transformed into an equivalent PHCA model.  An almost verbatim copy of its proof can be used to characterize those UCA models which are equivalent to some UHCA model.  However, this characterization can be done easily once Corollary~\ref{cor:PCA-models-of-PHCA} is known.

\begin{theorem}\label{thm:UHCA=PHCA+UCA}
 A graph is UHCA if and only if it is PHCA and UCA.  Moreover, every UCA model of a non-interval UHCA graph is UHCA.
\end{theorem}

\begin{proof}
 Clearly, every UHCA model is both PHCA and UCA.  For the converse, let $G$ be a PHCA and UCA graph and observe that $G$ contains no induced $K_{1,3}$.  If $G$ is also an interval graph, then it is a UIG graph by Theorem~\ref{thm:roberts}.  If $G$ is not an interval graph then, by Corollary~\ref{cor:PCA-models-of-PHCA}, every UCA model of $G$ is also HCA.
\end{proof}

The partial forbidden induced subgraph characterizations of UHCA graphs are shown below.

\begin{corollary}\label{cor:caracterizacion:UCA->UHCA}
 A UCA graph is UHCA if and only if it contains no induced $W_4$.
\end{corollary}

\begin{proof}
 It follows from Theorems \ref{thm:UHCA=PHCA+UCA}~and~\ref{thm:CInk}, Corollary~\ref{cor:caracterizacion-PCA->PHCA}, and the fact that $S_3 = CI(3,1)$.
\end{proof}

\begin{corollary}
 A PHCA graph is UHCA if and only if it contains no induced $CI(n,k)$ graph with $n > 3k$.
\end{corollary}

\begin{proof}
 By Theorem~\ref{thm:CInk}, UHCA graphs contain no induced $CI(n,k)$ for $n > 3k$.  For the converse, let $G$ be a PHCA graph having no induced $CI(n,k)$ with $n > 3k$.  By definition, every $CI(n,k)$ model with $2 < n < 3k$ has three arcs that together the circle.  Then, since $CI(n,k)$ graphs are not interval graphs, it follows by Corollary~\ref{cor:PCA-models-of-PHCA} that $CI(n,k)$ graphs with $2 < n < 3k$ are not PHCA.  So, in $G$ there is no induced $CI(n,k)$ for $n > 2k$.  Thus, Theorem~\ref{thm:CInk} and Corollary~\ref{cor:caracterizacion:UCA->UHCA} imply that $G$ is a UHCA graph.
\end{proof}

The whole picture of the CA class hierarchy is depicted in Figure~\ref{fig:hierarchy}.  Each box of the picture represents a subclass of circular-arc graphs.  An upright edge from the box corresponding to the class $\C_1$ to the box corresponding to the class $\C_2$ means that $\C_1$ is properly contained in $\C_2$.  Graphs that belong to $\C_2$ but not to $\C_1$ appear beside the edge corresponding to the inclusion of $\C_1$ in $\C_2$, except for the edge between the CA and the NCA classes since this family is unknown.  Finally, the label $\mathcal{O}$ beside the edge between the CA and HCA classes represents the family of obstacles that were defined by \citet{LinSzwarcfiter2006}.

\begin{figure}[ht!]
 \centering
 \includegraphics{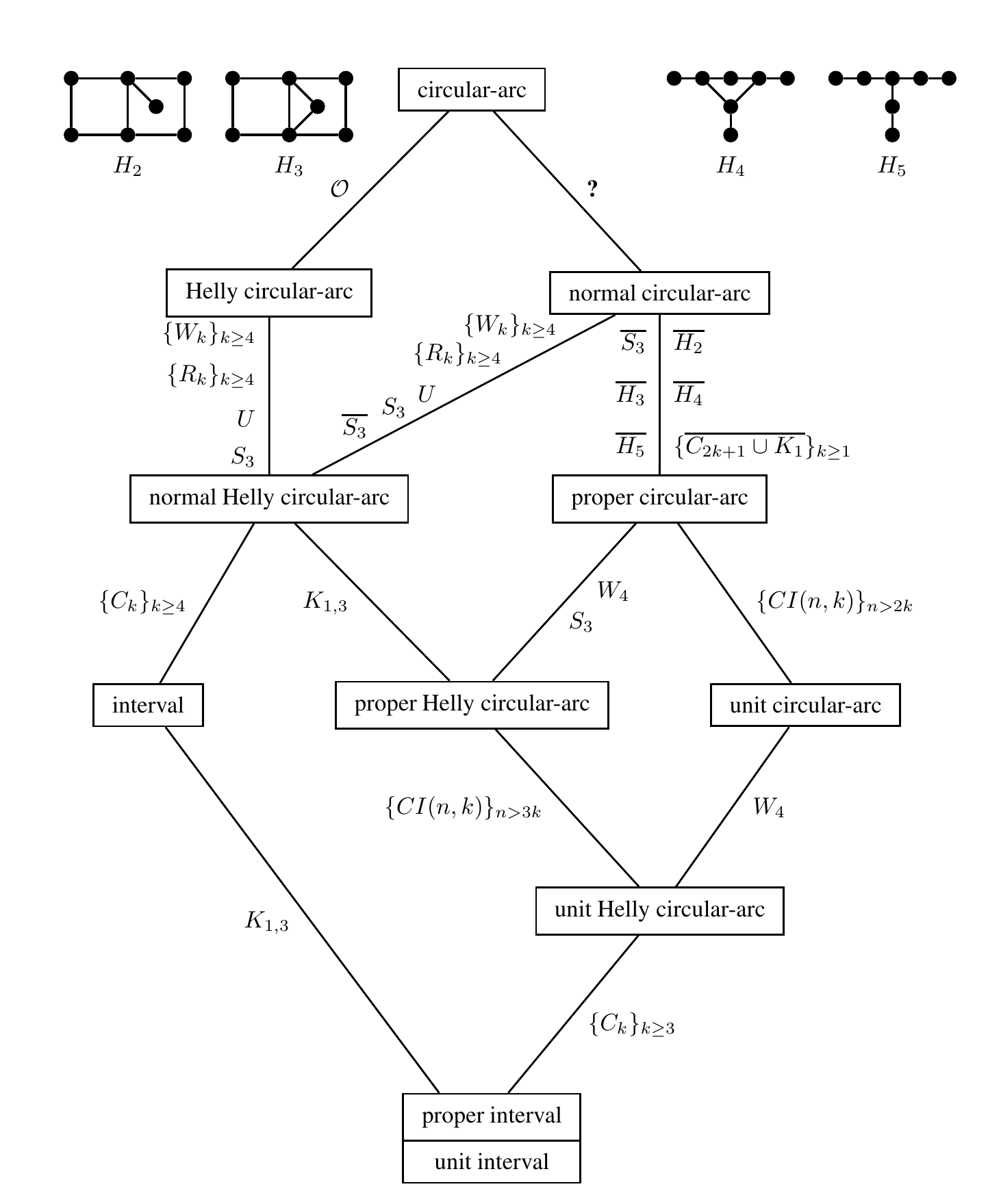}
 \caption{The class hierarchy of circular-arc graphs.}\label{fig:hierarchy}
\end{figure}

\section{The algorithms}
\label{sec:algorithms}

The characterizations described in the last section lead directly to $O(n+m)$ time algorithms for recognizing graphs in each of the NHCA subclasses.  This complexity is linear when the input is the graph represented by the adjacency lists.  So, the recognition problem for NHCA, PHCA, and UHCA graphs is well solved.  As for the model construction, first we refer to the more general class of circular-arc graphs.  Considering that the corresponding models of this class can be represented just by $O(n)$ elements, there is a motivation for trying to find $O(n)$ time algorithms that solve the recognition and model construction problems for each of the subclasses. In this case, the input is a general circular-arc model of a graph $G$ and the question is deciding whether $G$ belongs to a restricted class of circular-arc graphs and, whenever affirmative, constructing the corresponding restricted circular-arc model.  For example, given an arbitrary circular-arc model of a graph $G$, algorithms running in $O(n)$ time have been recently described to construct, whenever possible, a proper circular-arc model of $G$~\citep{Nussbaum2008}, a unit circular-arc model of $G$~\citep{LinSzwarcfiterSJDM2008}, or a Helly circular-arc model of $G$~\citep{JoerisLinMcConnellSpinradSzwarcfiterA2011}. 

In the present section we consider the representation problems for the classes of NHCA, PHCA, and UHCA graphs.  We propose new algorithms along the above lines and we discuss about the related certification and authentication problems.  We also propose a new $O(n)$ time algorithm that transforms any PCA model into an NPCA model (this algorithm is required by the PHCA and UHCA recognition algorithms).  Our algorithm is simpler than the previous known algorithms, both in its implementation and its time complexity analysis.

\subsection{Normalization of PCA models}
\label{sec:transformation:normalization}

\citet{TuckerDM1974} proved that every UCA graph $G$ with NPCA model $\M$ admits an UCA model $\M'$ in which all the extremes of $\M'$ appear in the same order as in $\M$.  This result is used by all the polynomial-time algorithms to recognize UCA graphs \citep[see][]{DuranGravanoMcConnellSpinradTuckerJA2006,KaplanNussbaumDAM2009,LinSzwarcfiterSJDM2008}, since they require the input to be an NPCA model.  

There are two algorithms for normalizing PCA models.  \citet{DuranGravanoMcConnellSpinradTuckerJA2006} showed how to normalize a PCA model in $O(n^2)$ time.  Their idea is to traverse each arc $A_i$ so as to find if there is some arc $A_j$ that together with $A_i$ cover the circle, \ie, $s(A_i), t(A_j), s(A_j), t(A_i)$ appear in this order in a traversal of the circle.  If such $A_j$ is found, then $t(A_j)$ and all the ending points inside $(s(A_i), t(A_j))$ are moved so that $t(A_j)$ and $s(A_i)$ appear in this order.  The main point is that every $A_j$ can be found and shrunk in $O(n)$ time, for each arc $A_i$.  \citet{LinSzwarcfiterSJDM2008} improved this algorithm so that all the pairs $A_i, A_j$ are found and shrunk in $O(n)$ time, obtaining an $O(n)$ time normalization algorithm.  The inconvenient of this new algorithm is that a somehow difficult analysis is used to show that every arc is shrunk at most twice.

Our elementary proof of Tucker's theorem yields a simple $O(n)$ time algorithm to transforms a PCA model into an NPCA model, depicted in Algorithm~\ref{alg:normalization}. The main loop computes the model $U_1(\M)$ and Step~\ref{alg:normalization:stepDup} inserts back all the possibly removed arcs.  The algorithm is correct by Theorem~\ref{thm:PCA->normal} and Lemma~\ref{lem:universal-arcs}.  With respect to the time complexity of the algorithm, Step~\ref{alg:normalization:step2} takes $O(1)$ time if each extreme has an index to its position in the traversal, and a pointer to the other extreme of the arc.  This information can be easily preprocessed in $O(n)$ time with one traversal of the model. Therefore, the total time complexity of the algorithm is $O(n)$.

\begin{algorithm}
\caption{Normalization procedure.}\label{alg:normalization}

\textbf{Input:} A PCA model $\M$ of a graph $G$.

\textbf{Output:} A twin-consecutive NPCA model of $G$.

\begin{AlgorithmSteps}
\Step{Set $u := 0$ and $U := NULL$.\label{alg:normalization:Step1}}
\Step{Traverse each arc $A$ of $\M$ and apply the following operations when $A$ contains at least $n-1$ extremes:\label{alg:normalization:step2}}
\IncreaseIndent
\Step{If $u = 0$, then set $u := 1$ and $U := A$.\label{alg:normalization:Step3}}
\Step{Otherwise, remove $A$ from $\M$ and set $u := u + 1$.\label{alg:normalization:Step4}}
\DecreaseIndent
\Step{If $u > 0$, then duplicate $u-1$ times the arc $U$ as in Lemma~\ref{lem:no_mellizos modelo}\label{alg:normalization:stepDup}.}
\end{AlgorithmSteps}
\end{algorithm}

Algorithm~\ref{alg:normalization} has three advantages over the one by \citet{LinSzwarcfiterSJDM2008}.  First, it is easier to implement and to describe.  Second, every extreme of the model is traversed only once, while for the algorithm in \citep{LinSzwarcfiterSJDM2008} there is only the guaranty that the main cycle is repeated at most $5n$ times.  Third, it turns out that PCA graphs admit only twin-consecutive models, except for the universal arcs, as Theorem~\ref{thm:PCA is twin consecutive} shows.  Thus, the model generated by Algorithm~\ref{alg:normalization} is twin-consecutive.

\begin{theorem}\label{thm:PCA is twin consecutive}
Every PCA model with at most one universal arc is twin-consecutive.
\end{theorem}

\begin{proof}
 Let $A_1$ and $A_2$ be a pair of twin arcs of a PCA model $\M$ and assume, w.l.o.g., that $s(A_1)$, $s(A_2)$, $t(A_1)$, and $t(A_2)$ appear in this order in a traversal of $C(\M)$.  To obtain a contradiction, suppose that $s(A_1)$ and $s(A_2)$ are not in the same $s$-sequence.  So, there is some ending point $t$, of some arc $A = (s,t)$, between $s(A_1)$ and $s(A_2)$.  Since $A_1$ is not universal then, by Lemma~\ref{lem:two covering arcs PCA}, it follows that $A_1$ and $A$ do not cover the circle.  Also, since $\M$ is proper, it follows that $s, s(A_i), t$ must appear in this order in a traversal of $C(\M)$.  Finally, since $A_1$ and $A_2$ are twins, then $A_2 \cap A \neq \emptyset$.  Summing up, $s(A_1), t, s(A_2), t(A_1), s, t(A_2)$ must appear in this order in $C(\M)$.  Since $A_1$ is not universal and $A_1$ and $A_2$ are twins, then there is some arc whose both extremes lie in $(t(A_2), s(A_1))$.  But this is a contradiction to the fact that $\M$ is proper, because this arc is contained in $A$.  Therefore, $s(A_1)$ and $s(A_2)$ lie in the same $s$-sequence.  Analogously, $t(A_1)$ and $t(A_2)$ lie in the same $t$-sequence.

 Now, suppose that there is some beginning point $s$ of an arc $A = (s,t)$ in the segment $(s(A_1), s(A_2))$. Then, $t$ must belong to the segment $(t(A_1), t(A_2))$ since otherwise $A$ is either contained in $A_1$ or containing $A_2$.  In other words, $A$ is a twin of both $A_1$ and $A_2$.  Analogously, if the ending point of some arc lies between $t(A_1)$ and $t(A_2)$, then its beginning point lies between $s(A_1)$ and $s(A_2)$.  Hence, every maximal set of twin arcs is consecutive in $\M$, \ie, $\M$ is twin-consecutive.
\end{proof}

\begin{corollary}\label{thm:NPCA twin-consecutive}
 The model generated by Algorithm~\ref{alg:normalization} is twin-consecutive.
\end{corollary}

\begin{proof}
 Immediately before the execution of Step~\ref{alg:normalization:stepDup}, $\M$ is a PCA model with at most one universal arc.  Then, by Theorem~\ref{thm:PCA is twin consecutive}, $\M$ is twin-consecutive at this point of the execution.  After the insertion of the $u-1$ copies of the universal arc as in Lemma~\ref{lem:no_mellizos modelo}, the beginning points as well as the ending points of the universal arcs are consecutive.  Therefore, the generated model is twin-consecutive.
\end{proof}

\subsection{Recognition of NHCA graphs}
\label{sec:algorithms:NHCA}

In this section we develop an algorithm to test whether a circular-arc graph is NHCA, when the input is a circular-arc model.  The algorithm follows directly from Corollary~\ref{cor:CA-models-of-NHCA} and its time complexity is $O(n+m)$.

Let $\M$ be a circular-arc model of a graph $G$.  If there are no two nor three arcs covering $C(\M)$, then $G$ is an NHCA graph and $\M$ is an NHCA model of $G$, so there is nothing to be done for this case.  Otherwise, it is enough to check if $\M$ is equivalent to an interval model since, by Corollary~\ref{cor:CA-models-of-NHCA}, those NHCA graphs that admit a non NHCA model are interval graphs.  If $\M$ is equivalent to an interval model, then $G$ is an NHCA graph and any interval model of $G$ is an NHCA model of $G$.  

We use Algorithm~\ref{alg:CA is NHCA} to test whether there are two or three arcs covering the circle.  There, $NEXT(A)$ represents the arc crossing $t(A)$ whose ending point reaches farthest, whereas $NEXT^2(A)$ represents $NEXT(NEXT(A))$.  In the second traversal of the first loop, $N$ is the arc reaching farthest among those crossing $e$, when $e$ is reached.  Thus, $NEXT$ is correctly computed in the first loop.  Now, observe that if $A_1$ and $A_2$ cover $C(\M)$, then $A_1$ and $NEXT(A_1)$ also cover $C(\M)$.  Similarly, if $A_1, A_2, A_3$ cover $C(\M)$, then $A_1, NEXT(A_1)$ and $NEXT^2(A_1)$ also cover $C(\M)$.  Thus, $\M$ is an NHCA model if and only if neither $NEXT(A)$ nor $NEXT^2(A)$ crosses $s(A)$, for every $A \in \M$.  Therefore, Algorithm~\ref{alg:CA is NHCA} is correct.

\begin{algorithm}
 \caption{Authentication of an NHCA model.}\label{alg:CA is NHCA}

 \textbf{Input:} A circular-arc model $\M$.

 \textbf{Output:} If $\M$ is not NHCA, then two or three arcs that cover $\M$.  Otherwise, there is no output.

 \begin{AlgorithmSteps}
  \Step{Let $A_1$ be any arc of $\M$ and set $A := A_1$.}
  \Step{Traverse $C(\M)$ twice from $s(A_1)$ and apply the following operation when an extreme $e$ of an arc $A$ is reached.}
  \IncreaseIndent
    \Step{If $e = s(A)$ and $t(A)$ reaches farther than $t(N)$, then set $N := A$.}
    \Step{If $e = t(A)$, then set $NEXT(A) := N$.}
  \DecreaseIndent

  \Step{Traverse $C(\M)$ once again from $s(A_1)$ and apply the following evaluation when $s(A)$ is reached.}
  \IncreaseIndent
    \Step{If $NEXT(A)$ crosses $s(A)$, then output $A, NEXT(A)$.\label{alg:CA is NHCA:Step unneeded in PCA is HCA}}
    \Step{If $NEXT^2(A)$ crosses $s(A)$, then output $A, NEXT(A)$ and $NEXT^2(A)$.}
 \end{AlgorithmSteps}
\end{algorithm}

With respect to the time complexity, all the operations of both loops take $O(1)$ time, thus the total time complexity of Algorithm~\ref{alg:CA is NHCA} is $O(n)$.

To test if $\M$ is equivalent to an interval model when $\M$ is not NHCA, we compute the intersection graph $G$ of $\M$ and we invoke the $O(n+m)$ time certified algorithm by \citet{KratschMcConnellMehlhornSpinradSJC2006}.  Unfortunately, we were unable to find an $O(n)$ time algorithm to obtain an interval model from $\M$.

We now discuss the certification and the authentication procedures.  When $\M$ is an NHCA model, the positive certificate is just $\M$.  If $\M$ is not an NHCA model, but it is equivalent to an interval model, the certificate is provided by the certified interval graph recognition algorithm in $O(n+m)$ time~\citep{KratschMcConnellMehlhornSpinradSJC2006}.  If $\M$ is neither NHCA nor equivalent to an interval model, then the negative certificate is obtained by combining the certificate of the interval graph recognition algorithm with the two or three arcs that cover the circle.  This certificate is enough by Corollary~\ref{cor:CA-models-of-NHCA}.  The negative certificate can be authenticated in $O(n)$ time as in \citep{KratschMcConnellMehlhornSpinradSJC2006}.  To authenticate the positive certificates it is enough to test that the output model $\M'$ is NHCA and equivalent to $\M$.  For the NHCA authentication use Algorithm~\ref{alg:CA is NHCA}, and for the isomorphism authentication use the $O(n)$ time algorithm by \citet{Curtis2007}. The complete certified procedure is summarized in Algorithm~\ref{alg:CA->NHCA}.

\begin{algorithm}
 \caption{Recognition of NHCA graphs.}\label{alg:CA->NHCA}

 \textbf{Input:} A circular-arc model $\M$.

 \textbf{Output:} Either an NHCA model equivalent to $\M$ or a subset of arcs whose intersection graph is not NHCA.

 \begin{AlgorithmSteps}
  \Step{Execute Algorithm~\ref{alg:CA is NHCA} to authenticate if $\M$ is an NHCA model.  If so, output $\M$.}
  \Step{Otherwise, execute the algorithm by \citet{KratschMcConnellMehlhornSpinradSJC2006} to the intersection graph $G$ of $\M$.  If $G$ is an interval graph, then output the interval model obtained by the algorithm.  Otherwise, output the two or three arcs covering $C(\M)$ together with the negative certificate obtained by the interval graph recognition algorithm.}
 \end{AlgorithmSteps}

\end{algorithm}

\subsection{Recognition of PHCA graphs}
\label{sec:algorithms:PHCA}

In this section we show two algorithms that can be used to recognize PHCA graphs.  The first one transforms an NHCA model into a PHCA model in $O(n)$ time, if possible.  The second one transforms a PCA model into a PHCA model in $O(n)$ time, if possible.  

To transform an NHCA model into a PHCA model we need only to sort the extreme sequences as in Theorem~\ref{thm:caracterizacion-NHCA->PHCA}.  If the model so obtained is not PHCA, then we can search for an induced $K_{1,3}$.  We begin by describing how to sort all the extreme sequences in $O(n)$ time.

Let $\M$ be an NHCA model.  For arcs $A_i, A_j \in \A(\M)$ with nonempty intersection, say that $s(A_i)$ \emph{appears before} $s(A_j)$ if $s(A_i)$ appears before $s(A_j)$ in a traversal of $C(\M)$ from some point $p \in C(\M) \setminus (A_i \cup A_j)$.  Observe that the point $p$ must always exist because $A_i$ and $A_j$ do not cover $C(\M)$.  Similarly, say that $t(A_i)$ \emph{appears before} $t(A_j)$ if $t(A_i)$ appears before $t(A_j)$ in a traversal of $C(\M)$ from the same point $p$.  The sorting of the extremes sequences in Theorem~\ref{thm:caracterizacion-NHCA->PHCA} can be rephrased as follows.  First, sort each $t$-sequence $T$ so that, for $t(A_i), t(A_j) \in T$, if $s(A_i)$ appears before $s(A_j)$, then $t(A_i)$ appears before $t(A_j)$.  Next, sort each $s$-sequence $S$ so that, for $s(A_i), s(A_j) \in S$, if $t(A_i)$ appears before $t(A_j)$, then $s(A_i)$ appears before $s(A_j)$.  The algorithms to sort the $t$-sequences and $s$-sequences are symmetric, so we only describe how to sort the $t$-sequences.

Let $T_1, \ldots, T_k$ be the $t$-sequences of $\M$ and $T$ be the set of all the ending points corresponding to arcs that cross some fixed beginning point $s$.  Consider the $t$-sequence $T_i'$ that results from sorting the $t$-sequence $T_i$, for some $1 \leq i \leq k$.  In $T_i'$, all the ending points of $T_i \cap T$ appear before all the ending points of $T_i \setminus T$.  Thus, we can sort all the $t$-sequences with four traversals of $C(\M)$ as in Algorithm~\ref{alg:sort NHCA extremes}.  In the first traversal of $C(\M)$, Algorithm~\ref{alg:sort NHCA extremes} marks all those arcs that cross the fixed beginning point $s(A_1)$.  Thus, the ending points of $T$ are precisely those ending points corresponding to the marked arcs.  The second traversal is used to find all the $t$-sequences $T_1, \ldots, T_k$ of $\M$.  The third traversal computes and sorts the sequences $T_{i,1} = T_i \cap T$ and $T_{i,2} = T_i \setminus T$, for every $1 \leq i \leq k$.  Note that, in Step~\ref{alg:sort NHCA extremes:step 9}, $t(A)$ is stored at the end of either $T_{i,1}$ or $T_{i,2}$, and all the ending points corresponding to arcs whose beginning point appears before $s(A)$ were already stored.  Finally, the last traversal of $C(\M)$ replaces each $t$-sequence with the sorted $t$-sequence.  Thus, the algorithm is correct.

\begin{algorithm}[!ht]
 \caption{Sorting of the $t$-sequences of an NHCA model.}\label{alg:sort NHCA extremes}

 \textbf{Input:} An NHCA model $\M$.

 \textbf{Output:} An NHCA model $\M'$ equivalent to $\M$ in which every $t$-sequences is sorted.  That is, if $t(A_i), t(A_j)$ are extremes of the $t$-sequence $T$ of $\M'$ and $s(A_i)$ appears before $s(A_j)$ then $t(A_i)$ appears before $t(A_j)$.

 \begin{AlgorithmSteps}
  \Step{Let $A_1$ be any arc of $\M$.}
    \Comment{Find the arcs that cross $s(A_1)$.}
  \Step{Traverse $\M$ from $s(A_1)$ and apply the following operation when an extreme $e$ of an arc $A$ is reached:}
  \IncreaseIndent
   \Step{If $e = s(A)$, then mark $A$.}
   \Step{If $e = t(A)$, then clear the mark of $A$.}
  \DecreaseIndent
  \Comment{Sort the extremes}
  \Step{Traverse $\M$ to compute the family $T_1, \ldots, T_k$ of $t$-sequences of $\M$.}
  \Step{For each $i := 1, \ldots, k$, define $T_{i,1}$ and $T_{i,2}$ as empty sequences.}
  \Step{Traverse $\M$ from $s(A_1)$ and apply the following each time a beginning point $s(A)$ is reached:}
  \IncreaseIndent
   \Step{Find the $t$-sequence $T_i$ that contains $t(A)$.}
   \Step{If $A$ is marked, insert $t(A)$ at the end of $T_{i,1}$; otherwise, insert $t(A)$ at the end of $T_{i,2}$.\label{alg:sort NHCA extremes:step 9}}
  \DecreaseIndent
  \Step{Replace $T_i$ with $T_{i,1}, T_{i,2}$ in $\M$ for every $1 \leq i \leq k$.}
  \Step{Output $\M$.}
 \end{AlgorithmSteps}
\end{algorithm}

With respect to the time complexity, all the operations of both loops take $O(1)$ time, while the computation of the $t$-sequences can be easily done in $O(n)$ time.  Therefore, the total time complexity of the sorting algorithm is $O(n)$.

After sorting the extremes, we must check whether its output model $\M'$ is PHCA or not.  Algorithm~\ref{alg:sort NHCA extremes} does not modify the elements that compose each $t$-sequence, thus $\M'$ is NHCA.  So, it is enough to check whether $\M'$ is PCA.  That is, we ought to check if the beginning points of the arcs appear in the same order as the ending points.  If affirmative, then $\M'$ is a PHCA model equivalent to the input model $\M$.  Otherwise, there are two arcs $A_i$ and $A_j$ such that $s(A_i), s(A_j), t(A_j)$ and $t(A_i)$ appear in this order in a traversal $C(\M')$.  Let $L$ be the arc whose ending point appears immediately after $s(A_i)$ and $R$ be the arc whose beginning point appears immediately before $t(A_i)$.  Arcs $A_i, A_j, L$, and $R$ are taken as the negative certificate since, as in the proof of Theorem~\ref{thm:caracterizacion-NHCA->PHCA}, they induce a circular-arc model of $K_{1,3}$.  As for the authentication, the negative certificate can be tested to be an induced submodel of $K_{1,3}$ in $\M$ in $O(1)$ time, if it is implemented as a set of four pointers.  To authenticate the positive certificate we ought to verify that the output model $\M'$ is normal, proper, Helly and equivalent to $\M$.  The NHCA properties can be checked with Algorithm~\ref{alg:CA is NHCA}, while the test of whether $\M'$ is PCA or not is done as in the PHCA recognition algorithm (Steps \ref{alg:NHCA->PHCA:auth1}--\ref{alg:NHCA->PHCA:auth2} of Algorithm~\ref{alg:NHCA->PHCA}).  Finally, if $\M'$ is PCA, then the equivalence of $\M$ and $\M'$ can be tested by running the PCA isomorphism algorithm by \citet{LinSoulignacSzwarcfiter2008}.  Algorithm~\ref{alg:NHCA->PHCA} summarizes the complete procedure.

\begin{algorithm}
 \caption{Recognition of PHCA graphs from NHCA models.}\label{alg:NHCA->PHCA}
 
 \textbf{Input:} An NHCA model $\M$.

 \textbf{Output:} Either a PHCA model equivalent to $\M$ or an induced submodel of $K_{1,3}$.
 
 \begin{AlgorithmSteps}
  \Step{Apply Algorithm~\ref{alg:sort NHCA extremes} to sort the $t$-sequences, and the symmetric algorithm to sort the $s$-sequences.}
  \Step{Let $A_1 := (s_1, t_1)$ be any arc of $\M$.}
  \Step{For $i := 1, \ldots, n$ do:\label{alg:NHCA->PHCA:auth1}}
  \IncreaseIndent
   \Step{Find the arc $A_{i+1} = (s_{i+1}, t_{i+1})$ whose beginning point is the first after $s_i$.}
   \Step{If $A_i \supset A_{i+1}$, then output $A_i, A_{i+1}$, the arc whose ending point appears first from $s_i$, and the arc whose beginning point appears first from $t_i$ in a counterclockwise traversal of $C(\M)$.\label{alg:NHCA->PHCA:auth2}}
  \DecreaseIndent
  \Step{Output $\M$.}
 \end{AlgorithmSteps}
\end{algorithm}

We now proceed to describe the algorithm that transforms a PCA model $\M$ into a PHCA model $\M'$, when possible.  The algorithm is a direct consequence of Theorem~\ref{thm:caracterizacion-PCA->PHCA}.  That is, it verifies if either $U_1(\M)$ is an HCA model or $U_0(\M)$ is an interval model.  If affirmative, then $\M$ is equivalent to a PHCA model, and one such PHCA model can be obtained as in Theorem~\ref{thm:caracterizacion-PCA->PHCA}.  Otherwise, the algorithm finds an induced submodel of $W_4$ or $S_3$.

The PHCA recognition algorithm is obtained by gluing together several parts of the algorithms developed so far.  The first step is to compute $U_1(\M)$ as in Steps \ref{alg:normalization:Step1}--\ref{alg:normalization:Step4} of Algorithm~\ref{alg:normalization}.  The second step is to verify whether $U_1(\M)$ is HCA.  For this, it is enough to invoke Algorithm~\ref{alg:CA is NHCA} so as to verify if $U_1(\M)$ is NHCA, because $U_1(\M)$ is NCA by Lemma~\ref{lem:two covering arcs PCA}.  However, we can simplify the computation of $NEXT$ so that it takes only one traversal of $C(\M)$.  Let $t(A_1), \ldots, t(A_k)$ be a $t$-sequence of $\M$.  Since $\M$ is PCA, then $NEXT(A_i)$ is the arc of $\M$ whose beginning point is closest to $t(A_i)$ in the counterclockwise direction.  Hence $NEXT(A_i) = NEXT(A_1)$, for every $1 \leq i \leq k$.  Therefore, with only one traversal we can find $NEXT(A)$ for every $A \in \A(\M)$.  The last step is to test whether $U_0(\M)$ is a PIG model or not, whenever $U_1(\M)$ has a universal arc $A$.  Instead of doing this, we can check that no beginning point appears before an ending point inside $A$, as it is done in Theorem~\ref{thm:caracterizacion-PCA->PHCA}.  All these steps can be implemented so as to run in $O(n)$ time with techniques similar to those discussed so far.

The algorithm described above can be modified so as to produce certificates in $O(n)$ time.  When $U_1(\M)$ is a PHCA model, we can obtain a PHCA model equivalent to $\M$ by including the universal arcs that where possible removed by Steps~\ref{alg:normalization:Step1}--\ref{alg:normalization:Step4} of Algorithm~\ref{alg:normalization}.  This can be achieved as in Step~\ref{alg:normalization:stepDup} of Algorithm~\ref{alg:normalization}.  To obtain the positive certificate when $U_0(\M)$ is a PIG model, we refer to the proof of Theorem~\ref{thm:caracterizacion-PCA->PHCA}, in particular, the implication $(iii) \Longrightarrow (i)$.  In this situation, $U_1(\M)$ contains a universal arc $A$. To obtain the required model, we include in ${\M}_0$ the arc $\overline{A}$, and then we can include the possible remaining universal arcs as in Step~\ref{alg:normalization:stepDup} of Algorithm~\ref{alg:normalization}.  The authentication of these certificates can be done in $O(n)$ time, as discussed for Algorithm~\ref{alg:NHCA->PHCA}.

The algorithm fails to transform the input PCA model into a PHCA model when $U_1(\M)$ is not HCA and $U_0(\M)$ is not PIG.  According to Theorem~\ref{thm:caracterizacion-PCA->PHCA}, an induced submodel of $\M$ whose intersection graph is either isomorphic to $W_4$ or to $S_3$ can be obtained as follows. Let $A_1, A_2$, and $A_3$ be the three arcs that together cover the circle of $U_1(\M)$. If none of these arcs is universal, then, as in Theorem~\ref{thm:caracterizacion-PCA->PHCA}, we know that there are three arcs $B_1,B_2,B_3$, such that $B_i$ intersects $A_j$ and not $A_i$, for all $1 \leq i, j \leq 3$, $i \neq j$.  In this case, the arcs $A_1, A_2, A_3, B_1, B_2$, and $B_3$ either induce a model of $S_3$ or contain a model of $W_4$.  On the other hand, if one among $A_1 ,A_2, A_3$, say $A_1$, is a universal arc then there are arcs $L, R$, such that $s(R)$ precedes $t(L)$ in $A_1$.  In this situation, a negative certificate can be obtained as above by replacing $A_1$ with either $R$ or $L$ when $R = A_2$ or $L = A_3$.  Finally, when $R \neq A_2$ and $L \neq A_3$, the arcs $A_1, A_2, A_3, L$ and $R$ induce the model of a forbidden $W_4$.  The authentication takes $O(1)$ time if the forbidden submodel is stored as a set of five or six pointers to the corresponding arcs of the model.  We summarize the above discussion in Algorithm~\ref{alg:PCA->PHCA}.

\begin{algorithm}[!ht]
 \caption{Recognition of PHCA graphs from PCA models.}\label{alg:PCA->PHCA}

 \textbf{Input:} A PCA model $\M$.

 \textbf{Output:} Either a PHCA model equivalent to $\M$ or an induced submodel of $W_4$ or $S_3$.

 \begin{AlgorithmSteps}
  \Step{Apply Steps~\ref{alg:normalization:Step1}--\ref{alg:normalization:Step4} of Algorithm~\ref{alg:normalization} to obtain $U_1(\M)$.\label{alg:PCA->PHCA:step1}}
  \Step{Execute Algorithm~\ref{alg:CA is NHCA} (simplified for PCA graphs), to verify whether $U_1(\M)$ contains three arcs that cover the circle. If negative, output the model obtained by the execution of Step~\ref{alg:normalization:stepDup} of Algorithm~\ref{alg:normalization} on $U_1(\M)$.}
  \Step{Let $A_1, A_2$ and $A_3$ be the three arcs of $U_1(\M)$ that were obtained in the previous step.}
  \Step{If $A_1, A_2$ and $A_3$ are not universal, then:}
  \IncreaseIndent
   \Step{Let $B_i$ be the arc whose beginning point is the first from $t(A_i)$, for $i \in \{1,2,3\}$.\label{alg:PCA->PHCA:step5}}
   \Step{If $B_1, B_2$ and $B_3$ are pairwise disjoint, then output $\{A_i, B_i\}_{1 \leq i \leq 3}$; otherwise, output $A_1, A_2, A_3$ and two intersecting arcs of $B_1, B_2$ and $B_3$.}
  \DecreaseIndent
  \Step{Sort $A_1, A_2$ and $A_3$ so that $A_1$ is the universal arc and $s(A_2) \in A_1$.}
  \Step{Traverse $\M$ from $s(A_1)$ to $t(A_1)$ to find if there are two arcs $L$ and $R$ such that $s(R)$ appears before $t(L)$. If $L$ and $R$ are found then:}
  \IncreaseIndent
   \Step{If $R = A_2$, then set $A_1 := R$ and goto Step~\ref{alg:PCA->PHCA:step5}.  If $L = A_3$, then set $A_1 := L$ and goto Step~\ref{alg:PCA->PHCA:step5}.  Output $A_1, A_2, A_3, L$ and $R$.}
  \DecreaseIndent
  \Step{Let $\M' := (\M \setminus \{A_1\}) \cup \{\overline{A_1}\}$.}
  \Step{Duplicate $\overline{A_1}$ in $\M'$ as many times as arcs where removed in Step~\ref{alg:PCA->PHCA:step1}, and return the model so obtained.}
 \end{AlgorithmSteps}
\end{algorithm}

\subsection{Recognition of UHCA graphs}
\label{sec:algorithms:UHCA}

In this section we briefly discuss how to transform either a UCA or a PHCA model $\M$ into a UHCA model, when possible.  We begin with the case in which $\M$ is PHCA. In this case, apply the algorithm by \citet{LinSzwarcfiterSJDM2008} to transform $\M$ into a UCA model $\M'$, if possible.  Since this algorithm preserves the order of the extremes of $\M$ then $\M'$ is both UCA and PHCA, \ie, $\M'$ is UHCA.  This algorithm takes $O(n)$ time and the model $\M'$ so generated can be authenticated to be UCA in $O(n)$ time.  If $\M$ is not equivalent to a UHCA model, then apply the algorithm by \citet{KaplanNussbaumDAM2009} to generate a negative certificate in $O(n)$ time.  This negative certificate can also be authenticated in $O(n)$ time.

Finally, consider the case in which $\M$ is a UCA model.  If $\M$ has two arcs that cover the circle, then the intersection graph of $\M$ is a complete graph and an equivalent UIG model is easily obtained in $O(n)$ time.  Suppose, then, that $\M$ is an NCA model, and consider the model $U_1(\M)$.  If $U_1(\M)$ is HCA then it is also UHCA and a UHCA model equivalent to $\M$ can be easily obtained in $O(n)$ time by duplicating the universal arc of $U_1(\M)$, if existing, as in Lemma~\ref{lem:no_mellizos modelo}.  If $U_1(\M)$ is not HCA but $U_0(\M)$ is an interval model then a PIG model equivalent to $\M$ can be obtained by applying Algorithm~\ref{alg:PCA->PHCA}.  A UIG model equivalent to $\M$ can be constructed in $O(n)$ time as in \citep{CorneilKimNatarajanOlariuSpragueIPL1995,LinSoulignacSzwarcfiter2009,Mitas1994}.  In the last case, if $U_1(\M)$ is not an HCA model and $U_0(\M)$ is not an interval model then, by Theorem~\ref{thm:caracterizacion-PCA->PHCA}, $\M$ is not equivalent to a UHCA model and a negative certificate is obtained as in Algorithm~\ref{alg:PCA->PHCA}.  All the certificate authentications take $O(n)$ time as before.

\subsection{Summary of the transformation algorithms}

In Table~\ref{tab:algorithms} we summarize the complexities of the transformation algorithms between circular-arc subclasses.  The main open problems are how to transform a circular-arc model into an NCA model in polynomial time, and how to transform a circular-arc model into an NHCA model in $O(n)$ time.  We recall that this last problem can be reduced to transforming an HCA model into an interval model in $O(n)$ time.

\begin{table}[ht!]
 \centering
 \begin{tabular}{|l|l|c|c|}
  \hline
  From                     & To        & Time complexity   & References \\
  \hline
  CA                       & NCA       & open              &  \\
  co-bipartite CA          & NCA       & $O(n^5m^6\log m)$ & \citet{MullerDAM1997,HellHuangJGT2004} \\
  CA                       & HCA       & $O(n)$            & \citet{JoerisLinMcConnellSpinradSzwarcfiterA2011} \\
  CA $\cup$ NCA            & PCA       & $O(n)$            & \citet{Nussbaum2008} \\
  PCA                      & UCA       & $O(n)$            & \citet{LinSzwarcfiterSJDM2008} \\
  CA $\cup$ HCA $\cup$ NCA & NHCA      & $O(n+m)$          & \citet{KratschMcConnellMehlhornSpinradSJC2006} \& \S~\ref{sec:algorithms:NHCA}  \\
  NHCA                     & PHCA      & $O(n)$            & \S\ref{sec:algorithms:PHCA} \\
  PCA                      & PHCA      & $O(n)$            & \S\ref{sec:algorithms:PHCA} \\
  UCA                      & UHCA      & $O(n)$            & \S\ref{sec:algorithms:PHCA} \& \S\ref{sec:algorithms:UHCA} \\
  PHCA                     & UHCA      & $O(n)$            & \citet{LinSzwarcfiterSJDM2008} \& \S\ref{sec:algorithms:PHCA} \\
  NHCA                     & IG        & $O(n)$            & Corollary~\ref{cor:CA-models-of-NHCA} \\
  PHCA $\cup$ UHCA         & PIG       & $O(n)$            & \S\ref{sec:algorithms:PHCA} \\
  IG                       & PIG       & $O(n)$            & \S\ref{sec:algorithms:NHCA} \\
  PIG                      & UIG       & $O(n)$            & \parbox[c]{6.5cm}{\centering \citet{CorneilKimNatarajanOlariuSpragueIPL1995,LinSoulignacSzwarcfiter2009}; \\ \citet{Mitas1994}} \\
  \hline
 \end{tabular}
 \caption[Time complexities of the transformation algorithms]{Time complexities of the transformation algorithms.  All algorithms proposed in this chapter are certified.}\label{tab:algorithms}
\end{table}

\section{Some additional properties of the NHCA subclasses}
\label{sec:properties}

In this section we prove some properties about the NHCA subclasses that might be useful from an algorithmic point of view.  Some of these were presented in Section~\ref{sec:motivation} as natural generalizations of properties about interval graphs.  

\subsection{Counting NPHCA models of PHCA graphs}

The first problem is to count how many NPHCA models does a PHCA graph admit.  This question of how many representations of a graph are there has been solved for both PIG graphs and co-connected PCA graphs; \citet{Roberts1969} proved that the PIG model of a connected PIG graph is unique up to full reversal, while \citet{HuangJCTSB1995} proved that the PCA model of a connected and co-connected PCA graph is unique up to full reversal.  

\begin{theorem}[\citealp{Roberts1969}]\label{thm:unique-PIG-models}
 Every connected PIG graph admits at most two PIG models, one the reverse of the other.
\end{theorem}

\begin{theorem}[\citealp{HuangJCTSB1995}]\label{thm:unique-PCA-models}
 Every connected PCA graph whose complement is either connected or non-bipartite admits at most two PCA models, one the reverse of the other.
\end{theorem}

Besides the theoretical interest behind these questions, it turns out that the uniqueness of PIG models is strongly used to solve the recognition problem.  In fact, interval graphs can be recognized in $O(1)$ time per edge insertion due to this property \citep[see][]{DengHellHuangSJC1996,HellShamirSharanSJC2001}.  Because of the strong relationship between PIG and PHCA graphs, it should come at no surprise that there is a unique NPHCA model of every PHCA graph, up to full reversal.  Even more, this property can be exploited so as to generalize the PIG recognition algorithms to the PHCA case, with no much effort \citep[see][]{Soulignac2010}.

\begin{theorem}\label{thm:unique-PHCA-models}
 Every connected PHCA graph admits at most two NPHCA models, one the reverse of the other.
\end{theorem}

\begin{proof}
 Suppose, for the base case, that $G$ contains no pair of twin vertices.  If $\overline{G}$ is connected or it is non-bipartite then the result follows from Theorem~\ref{thm:unique-PCA-models}.  Then, it is enough to deal with the case in which $\overline{G}$ has $k > 1$ components $\overline{H_1}, \ldots, \overline{H_k}$, all of which are bipartite.  We denote by $H_i$ the subgraph of $G$ induced by the vertices of $\overline{H_i}$, for $1 \leq i \leq k$.  Since $G$ has no twin vertices and it is $W_4$-free by Corollary~\ref{cor:caracterizacion-PCA->PHCA}, then it follows that $k = 2$.  Analyze the following two cases:

\begin{description}
 \item [Case 1:] $|H_1| \geq |H_2| > 1$.  If $H_1$ contains an induced path of four vertices $v_1, v_2, v_3, v_4$ then $v_1, v_2, v_3$ together with a pair of non-adjacent vertices of $H_2$ induce a $W_4$ in $G$, a contradiction to Corollary~\ref{cor:caracterizacion-PCA->PHCA}.  Then, $\overline{H_1}$ is bipartite and $H_1$ contains no twins and no paths of four vertices.  Thus, $H_1$ if isomorphic to $\overline{P_2}$.  Then, by the case hypothesis, $H_2$ is also isomorphic to $\overline{P_2}$.  Consequently, $G$ is isomorphic to $C_4$, which admits a unique circular-arc model.

 \item [Case 2:] $|H_2| = 1$.  In this case $G$ contains no hole, or otherwise the hole and the vertex of $H_2$ would induce a wheel. Let $\M$ be an NPHCA model of $G$. Since $G$ contains no hole, it follows that some point of $C(\M)$ is not covered by the arcs of $\M$.  In other words, every NPHCA model of $G$ is a PIG model.  Hence, by Theorem~\ref{thm:unique-PIG-models}, $G$ admits two NPHCA models, one the reverse of the other.
\end{description}

 For the inductive case, observe that either $G$ contains no universal vertices or every NPHCA model of $G$ is PIG as in Case 2 above.  In the former case, every NPHCA model of $G$ is twin-consecutive by Theorem~\ref{thm:PCA is twin consecutive}, thus there is a unique model by the inductive hypothesis.  In the latter case, $G$ admits a unique PIG model by Theorem~\ref{thm:unique-PIG-models}.
\end{proof}

\subsection{Circular-ones properties of the clique matrix}

For the second part of this section, we study the relationship between PHCA graphs and the circular-ones properties of the clique matrix.  A $0$-$1$ matrix $M$ has the \emph{consecutive-ones property for rows} if its columns can be ordered so that, in every row, the ones are consecutive.  Matrix $M$ has the \emph{circular-ones property for rows} if the columns can be ordered so that, in every row, either the zeros or the ones are consecutive.  The consecutive and circular-ones properties for columns are defined analogously.  That is, $M$ has the \emph{consecutive-ones property for columns} if $M^T$ has the consecutive-ones property for rows, while $M$ has the \emph{circular-ones property for columns} if $M^T$ has the circular-ones property for rows.  Here $M^T$ is the transpose matrix of $M$.

Let $C_1, \ldots, C_k$ be the cliques of a graph $G$ and $v_1, \ldots, v_n$ be its vertices.  The \emph{clique-vertex incidence matrix} of $G$, or simply the \emph{clique matrix} of $G$, is the $0$-$1$ matrix $Q(G)$ with $k$ rows and $n$ columns where $Q(G)_{i,j} = 1$ if and only if vertex $v_j$ belongs to clique $C_i$, for every $1 \leq i \leq k$, $1 \leq j \leq n$.  

We mentioned in Section~\ref{sec:motivation} that $Q(G)$ has the consecutive-ones (resp.\ circular-ones) property for columns if and only if $G$ is an interval graph (resp.\ HCA graph).  The stronger condition of $Q(G)$ having also the consecutive-ones property for rows is equivalent to the condition of $G$ being a proper interval graph.  We prove the analogous result for the circular-ones property, \ie, $Q(G)$ has the circular-ones property for both rows and columns if and only if $G$ is a PHCA graph.

\begin{theorem}
 A graph $G$ is a PHCA graph if and only if $Q(G)$ has the circular-ones property for both rows and columns.
\end{theorem}

\begin{proof}
 Let $\M$ be a PHCA model of $G$ and let $A_1, \ldots, A_n$ be the arcs of $\M$ in order of appearance of its beginning points.  Since $\M$ is HCA, it follows that each clique is represented by some clique point.  Let $p_1, \ldots, p_k$ be the clique points in circular order.  Define $Q$ as the $k \times n$ matrix where $Q_{i,j} = 1$ if $A_j$ crosses $p_i$, and $0$ otherwise.  By definition, $Q$ is a clique matrix of $G$.  Since we used the same construction as \citet{GavrilN1974}, it follows that $Q$ has the circular-ones property for columns.  We now show that $Q$ has also the circular-ones property for rows.  Let $r$ be some row of $Q$ and represent by $r_i$ the $i$-th column of $r$.  Suppose that the ones in $r$ are not all consecutive.  Then, there exist $a, b, c$ such that $r_a = r_c = 1$, $r_b = 0$ and $1 \leq a < b < c \leq n$.  In other words, the clique point $p_i$ is crossed by $A_a$ and $A_c$, but not by $A_b$.  Since $s(A_a), s(A_b), s(A_c)$ appear in this order in $\M$ and $\M$ is proper, then $s(A_a), p_i, t(A_c), s(A_b), t(A_b)$ appear in this order in $\M$.  Even more, since $\M$ is proper, it follows that $A_d$ crosses $p_i$ for every $d$ such that $1 \leq d \leq  a$ or $c \leq d \leq n$.  Thus, every zero in $r$ is consecutive, and so $Q$ has the circular-ones property for rows.

 For the converse, we show that if $G$ is not a PHCA graph, then $Q(G)$ does not have the circular-ones property for either the rows or the columns.  If $G$ is not an HCA graph, then $Q(G)$ does not have the circular-ones property for columns \citep{GavrilN1974}.  Suppose then that $G$ is HCA and it is not PHCA.  Then, by Corollary~\ref{cor:caracterizacion-HCA->NHCA} and Theorem~\ref{thm:caracterizacion-NHCA->PHCA}, $G$ contains a $K_{1,3}$, a $W_4$, a $W_5$, or an $S_3$ as an induced subgraph.  Clique matrices for these graphs are depicted in Figure~\ref{fig:clique-matrices-of-nonPHCA}.  By inspection, none of these matrices has the circular-ones property for rows.  Thus, as $Q(G)$ contains at least one (permutation) of these matrices as a submatrix, it follows that $Q(G)$ does not have the circular-ones property for rows.
\end{proof}

\begin{figure}
 \hspace*{\stretch{1}}
 \parbox{29mm}{%
\begin{displaymath}%
 \left(
 \begin{array}{cccc}
  1 & 1 & 0 & 0 \\
  1 & 0 & 1 & 0 \\
  1 & 0 & 0 & 1 
 \end{array}
 \right)
\end{displaymath}%
 }%
 \hspace*{\stretch{1}}%
\parbox{34mm}{%
\begin{displaymath}%
 \left(
 \begin{array}{ccccc}
  1 & 1 & 0 & 0 & 1 \\
  1 & 1 & 1 & 0 & 0 \\
  1 & 0 & 1 & 1 & 0 \\
  1 & 0 & 0 & 1 & 1
 \end{array}
 \right)
\end{displaymath}%
 }%
 \hspace*{\stretch{1}}%
\parbox{40mm}{%
\begin{displaymath}%
 \left(
 \begin{array}{cccccc}
  1 & 1 & 0 & 0 & 0 & 1 \\
  1 & 1 & 1 & 0 & 0 & 0 \\
  1 & 0 & 1 & 1 & 0 & 0 \\
  1 & 0 & 0 & 1 & 1 & 0 \\
  1 & 0 & 0 & 0 & 1 & 1
 \end{array}
 \right)
\end{displaymath}%
 }%
 \hspace*{\stretch{1}}%
\parbox{40mm}{%
\begin{displaymath}%
 \left(
 \begin{array}{cccccc}
  1 & 0 & 1 & 0 & 1 & 0 \\
  1 & 1 & 0 & 0 & 0 & 1 \\
  1 & 1 & 1 & 0 & 0 & 0 \\
  0 & 1 & 1 & 1 & 0 & 0 \\
 \end{array}
 \right)
\end{displaymath}%
 }%
 \hspace*{\stretch{1}}
 \caption[Clique matrices of $K_{1,3}$, $W_4$, $W_5$ and $S_3$]{From left to right, the clique matrices of $K_{1,3}$, $W_4$, $W_5$, and $S_3$ are shown.}\label{fig:clique-matrices-of-nonPHCA}
\end{figure}

\subsection{The local interval property for graphs}
\label{sec:orientations}

Recall that, by the local interval property, $\M$ is an NHCA model if and only if $N[A]$ is an interval model, for every arc $A \in \A(\M)$.  As a consequence, we obtain that if $G$ is an NHCA graph, then $N[v]$ is an interval graph for every $v \in V(G)$.  The converse does not hold; the umbrella is not an NHCA graph, but the closed neighborhood of each of its vertices is an interval graph.  Nevertheless, it is possible to give a necessary and sufficient condition that reflects the fact that NHCA graphs are local interval, independent of how their models look like.  This local interval property about graphs is based on straight and round oriented graphs.  Unlike the previous properties described in this paper, these characterizations can be generalized to NCA and PCA graphs and, for this reason, they fill a gap by showing exactly what is lost in the jump from interval graph to NCA graphs and from PIG graphs to PCA graphs \citep[][present many classes of graphs defined as those graphs that can be oriented with some particularity]{Bang-JensenGutin2001}.  We begin by introducing some useful notation and terminology. 

For an ordering $X = x_1, \ldots, x_n$, we say that $x_i$ is to the \emph{left} of $x_j$ and $x_j$ is to the \emph{right} of $x_i$, for every $1 \leq i < j \leq n$.  The elements $x_1$ and $x_n$ are respectively called the \emph{leftmost} and \emph{rightmost} of $X$.  For $x_i, x_j \in X$, the \emph{circular range} $\cir[x_i, x_j]$ is defined as the ordering $x_i, x_{i+1}, \ldots, x_j$ where, as usual, all the operations are taken modulo $n$.  Similarly, the range $\cir[x_i, x_j)$ is obtained by removing the rightmost element from $\cir[x_i, x_j]$, the range $\cir(x_i, x_j]$ is obtained by removing the leftmost element from $\cir[x_i, x_j]$, and $\cir(x_i, x_j)$ is obtained by removing both the leftmost and rightmost elements from $\cir[x_i, x_j]$.  For $1 \leq i \leq j \leq n$, define the \emph{linear range} $\lin[x_i, x_j]$ as $[x_i, x_j]$, while for $i > j$ define $\lin[x_i, x_j]$ as the empty ordering.  The linear ranges $\lin[x_i, x_j)$, $\lin(x_i, x_j]$ and $\lin(x_i, x_j)$ are defined analogously.

An oriented graph $D$ is \emph{out-straight} if there is an ordering $v_1, \ldots, v_n$ of $V(D)$ such that, for every vertex $v_i$,  $N^+[v_i] =  \lin[v_{i}, v_{i+r}]$, where $r = d^+(v_i)$.  The ordering $v_1, \ldots, v_n$ is referred to as an \emph{out-straight enumeration} of $D$.  It is not hard to see that a graph $G$ is an interval graph if and only if it admits an out-straight orientation $D$.  Indeed, $\{(s_i,t_i) \mid v_i \in V(G)\}$ is an interval model of a graph $G$ such that $s_1 < \ldots < s_n$ if and only if $v_1, \ldots, v_n$ is an out-straight enumeration of the digraph with vertex set $\{v_1, \ldots, v_n\}$ and edge set $\{v_i \to v_j \mid i < j \text{ and } t_i > s_j\}$.  In Figure~\ref{fig:out-straight}, an interval graph, together with one of its out-straight orientations, is depicted.

\begin{figure}
 \centering
 \includegraphics{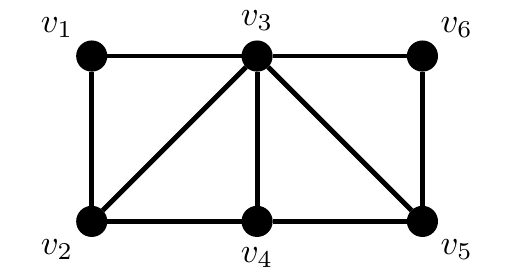} \hspace*{1cm} \includegraphics{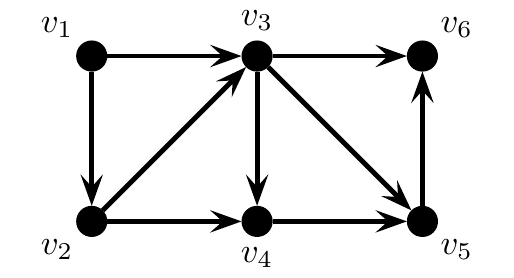}  
 \caption{An interval graph and one of its corresponding out-straight orientations.  The ordering $v_1, \ldots, v_6$ is an out-straight enumeration of the orientation.}\label{fig:out-straight}
\end{figure}

The most common way to generalize the definition of out-straight enumerations is to exchange the linear range $\lin[\bullet, \bullet]$ with a circular range $\cir[\bullet, \bullet]$.  That is, an oriented graph $D$ is \emph{out-round} if there is an ordering $v_1, \ldots, v_n$ of $V(D)$ such that, for every vertex $v_i$, $N^+[v_i] = \cir[v_i, v_{i+r}]$, where $r = d^+(v_i)$.  The ordering $v_1, \ldots, v_n$ is referred to as an \emph{out-round enumeration} of $D$.  Figure~\ref{fig:out-round} shows three examples of out-round oriented graphs.

Although out-round oriented graphs are a natural generalization of out-straight oriented graphs, there is a key property about the scopes of the vertices that is completely lost.  Let $\phi = v_1, \ldots, v_n$ be an ordering of $V(D)$ for an oriented graph $D$.  For $v_i \in V(D)$, define the \emph{leftmost} neighbor of $v_i$ in $\phi$ as the vertex $v_l \in N^-[v]$ that appears last when traversing $\phi$ from $v_i$ in reverse order.  Similarly, define the \emph{rightmost} neighbor of $v_i$ in $\phi$ as the vertex $v_r \in N^+[v]$ that appears last when traversing $\phi$ from $v_i$ in forward order.  The \emph{scope} of $v_i$ in $\phi$ is the range $\cir[v_l, v_r]$ where $v_l$ and $v_r$ are the leftmost and rightmost neighbors of $v_i$ in $\phi$ (see Figure \ref{fig:out-round}~(c)).  The scope of $v_i$ can be thought as the unique range $S = \cir[v_l, v_r]$ such that $N[v_i] \subseteq S$, $v_l \in N^-[v_i]$ and $v_r \in N^+[v_i]$ (the inconvenient is that such range $S$ is not well defined when $v_i \not \in \cir[v_l, v_r]$).

It is easy to see that if $\phi = v_1, \ldots, v_n$ is an out-straight enumeration of an oriented graph $D$, then the scope $S$ of the vertex $v_i$ is an out-straight enumeration of $D[S]$, for every $v_i \in V(D)$.  This property does not hold for out-round oriented graphs.  The oriented $4$-wheel graph $D$ in Figure \ref{fig:out-round}~(a) is out-round, but the scope $S$ of its universal vertex contains an induced hole and, therefore, $D[S]$ is not out-straight.  We define the locally out-straight oriented graphs specifically to restore this property back.  That is, an oriented graph $D$ is \emph{locally out-straight} if there is an out-round enumeration $\phi = v_1, \ldots, v_n$ of $D$ such that $S$ is an out-straight enumeration of $D[S]$, for every $v_i \in V(D)$ with scope $S$ in $\phi$.  As before, the enumeration $\phi$ is referred to as a \emph{locally out-straight enumeration} of $D$.  The graph depicted in Figure \ref{fig:out-round}~(b) is locally out-straight and it is not out-straight.

\begin{figure}
 \centering
 \begin{tabular}{c@{\hspace{5mm}}c@{\hspace{5mm}}c}
 \includegraphics{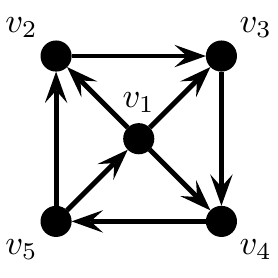} & \includegraphics{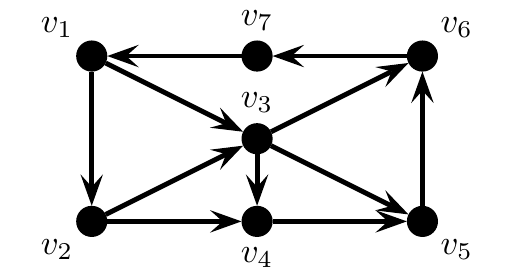} & \includegraphics{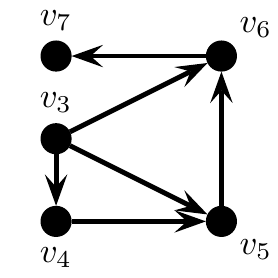}\\
 (a) & (b) & (c)
 \end{tabular}
 \caption[Examples of round oriented graphs]{Examples of round oriented graphs.  An out-round orientation of the $4$-wheel graph is depicted in~(a). Figure~(b) shows an orientation of a non-interval graph together with a locally out-straight enumeration $\phi = v_1, \ldots, v_7$.  The scope of $v_6$ in $\phi$ appears at~(c).}\label{fig:out-round}
\end{figure}

The next two theorems relate the out-round and the locally out-straight oriented graphs with the NCA and NHCA graphs, respectively.

\begin{theorem}\label{thm:NCA = right-consecutive}
 A graph is an NCA graph if and only if it admits an out-round orientation.
\end{theorem}

\begin{proof}
 Let $\M$ be an NCA model of a graph $G$ with arcs $A_1, \ldots, A_n$, where $s(A_1)$, $\ldots$, $s(A_n)$ appear in this order in a traversal of $C(\M)$.  For $1 \leq i \leq n$, call $v_i$ to the vertex of $G$ that corresponds to $A_i$.  Define $D$ as the digraph with vertex set $V(G)$, where $v_i \to v_j$ if and only if $A_i$ crosses $s(A_j)$, for every $1 \leq i,j \leq n$.  We claim that $D$ is an out-round orientation of $G$ with $\phi = v_1, \ldots, v_n$ as an out-round enumeration.  Fix $i$ and $j$ such that $1 \leq i,j \leq n$.  First, notice that $s(A_i) \in A_j$ only if $s(A_j) \not \in A_i$ because $\M$ is NCA.  Thus, $v_i \to v_j$ if and only if $v_i$ is adjacent to $v_j$ and $v_j \nto v_i$, \ie, $D$ is an orientation of $G$.  Second, if $v_i \to v_j$ then, since $s(A_j) \in A_i$, it follows that $v_i \to v_k$ for every $v_k \in \cir(v_i, v_j]$.  Hence $N^+[v_i] = \cir[v_i, v_{i+r}]$, where $r = d^+(v_i)$, implying that $\phi$ is an out-round enumeration of $D$.

 For the converse, let $\phi = v_1, \ldots, v_n$ be an out-round enumeration of some orientation $D$ of $G$.  Pick $n$ points $s(1), \ldots, s(n)$ of a circle $C$ in such a way that $s(1), \ldots, s(n)$ appear in this order in a traversal of $C$.  For each vertex $v_i \in V(D)$, define $A_i$ as the arc of $C$ whose beginning point is $s(i)$ and whose ending point lies in $(s(i + r_i), s(i + r_i + 1))$, where $r_i = d^+(v_i)$.  We claim that $\M = (C, \{A_i\}_{1 \leq i \leq n})$ is an NCA model of $G$.  Fix $i$ and $j$ such that $1 \leq i, j \leq n$.  By definition, $v_i \to v_j$ if and only if $i < j \leq i + r_i$, thus $A_i$ crosses $s(j) = s(A_j)$ if and only if $v_i \to v_j$, \ie, $\M$ is a circular-arc model of $G$.  On the other hand, if $v_i \to v_j$, then $v_j \nto v_i$, hence if $A_i$ crosses $s(A_j)$, then $A_j$ does not cross $s(A_i)$.  That is, $\M$ is an NCA model of $G$.
\end{proof}

\begin{theorem}\label{thm:locally-interval}
 A graph is an NHCA graph if and only if it admits a locally out-straight orientation.
\end{theorem}

\begin{proof}
 Let $\M$ be an NHCA model of a graph $G$ with arcs $A_1, \ldots, A_n$, where $s(A_1)$, $\ldots$, $s(A_n)$ appear in this order in a traversal of $C(\M)$.  Define the set $\{v_i\}_{1 \leq i \leq n}$, the orientation $D$ of $G$, and the out-straight enumeration $\phi$ of $D$ as in the proof of Theorem~\ref{thm:NCA = right-consecutive}.  We claim that $D$ is actually a locally out-straight orientation of $G$ with $\phi$ as a locally out-straight enumeration.  Fix $i$ such that $1 \leq i \leq n$, and take a small enough $\epsilon$.  Let $A_l$ be the arc crossing $s_i+\epsilon$ whose beginning point is farthest from $s_i$ in a counterclockwise traversal of $C(\M)$.  Similarly, let $A_r$ be the arc crossing $t_i-\epsilon$ whose beginning point is nearest to $t_i$ in a counterclockwise traversal of $C(\M)$.  By the way $\phi$ is defined, it follows that $v_l$ and $v_r$ are the leftmost and rightmost neighbors of $v_i$ in $\phi$.  Hence, $S = \cir[v_l, v_r]$ is the scope of vertex $v_i$ in $\phi$.  

 Since $s(A_1), \ldots, s(A_n)$ appear in this order in $\M$, it follows that $v_j \in S$ if and only if $s(A_j) \in (s_l-\epsilon, t_i+\epsilon)$.  Thus, $A_j$ cannot cross $s_l$, or otherwise $A_j$ together with $A_l$ and $A_i$ would cover the circle.  Therefore, the submodel $\M'$ of $\M$ induced by the arcs $A_l, \ldots, A_r$ is an interval model of $G[S]$.  This implies that $S$ is a locally out-straight enumeration of $D[S]$.  Consequently, $\phi$ is a locally out-straight enumeration of $D$ as claimed.

 For the converse, let $G$ be a graph that admits a locally out-straight orientation.  By Theorem~\ref{thm:NCA = right-consecutive}, $G$ must be an NCA graph.  So, it is enough to prove that $G$ contains no wheels, $3$-suns, umbrellas, rising suns, nor tents as induced subgraphs, by Corollary~\ref{cor:caracterizacion-NCA->NHCA}.  It is not hard to see that the class of graphs that admit a locally out-straight orientation is hereditary, hence we need only prove that wheels, $3$-suns, umbrellas, rising suns, and tents admit no locally out-straight orientations.  We shall do this for the proof.

 For the first case, let $D$ be an out-round orientation of a wheel and take $\phi$ as an out-round enumeration of $D$.   Clearly, if $v$ is the universal vertex of the wheel and $S$ is its scope, then $N[v] = V(D) \subseteq S$.  Consequently, $D[S]$ is not out-straight because it contains a hole, implying that $\phi$ is not a locally out-straight enumeration.  Therefore, $D$ is not locally out-straight.

 For the second case, let $D$ be an out-round orientation of the $n$-rising sun graph, for $n \geq 4$.  Recall that the $n$-rising sun is the graph obtained by inserting two universal vertices $v_1$ and $v_n$ into a path $P = v_2, \ldots, v_{n-1}$, and then inserting three vertices $w_1, w_{n-1}, w_n$ such that $w_i$ is adjacent only to $v_{i}$ and $v_{i+1}$, for $i \in \{1,n-1,n\}$ (see Figure~\ref{fig:NHCA = los} (a)).  Suppose, to obtain a contradiction, that $D$ admits a locally out-straight enumeration $\phi$.  It is not hard to see that the induced path $w_1, P, w_{n-1}$ admits only the following out-round enumerations, all of which are out-straight: 
 \begin{itemize}
  \item $\rho = w_1, v_2, \ldots, v_{n-1}, w_{n-1}$,
  \item $\gamma = v_2, w_1, v_3, \ldots, v_{n-1}, w_{n-1}$,
  \item $\rho' = w_{n-1}, v_{n-1}, \ldots, v_2, w_{1}$, and
  \item $\gamma' = v_{n-1}, w_{n-1}, v_{n-2}, \ldots, v_2, w_1$.
 \end{itemize}
  
 \begin{figure}
  \begin{tabular*}{\textwidth}{@{\extracolsep{\fill}}*3c}
   \includegraphics{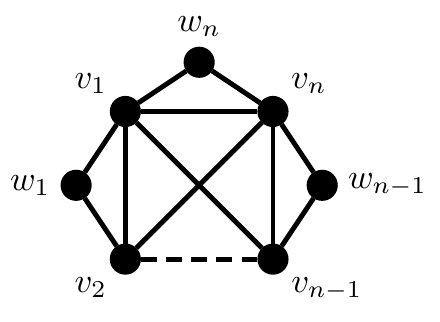} &  \includegraphics{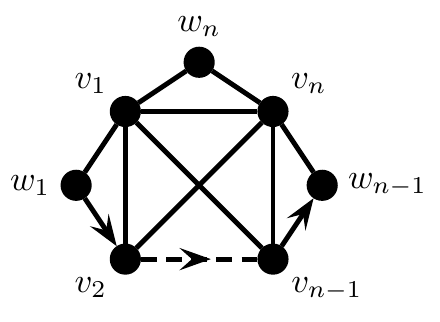} & \includegraphics{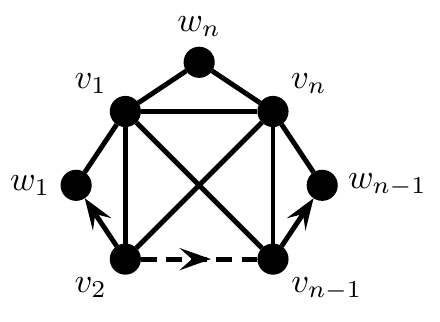} \\
    (a) & (b) & (c)
  \end{tabular*}
  \caption{The $n$-rising sun graph and the orientations corresponding to $\rho$ and $\gamma$.}\label{fig:NHCA = los}
 \end{figure}  
  
 Each of these enumerations corresponds to one of the four possible out-straight orientations of a path (see Figures \ref{fig:NHCA = los}~(b)~and~(c)).  So, one of these enumerations must appear in this order inside $\phi$; call it $\phi'$.  Furthermore, a vertex $z_1$ in $\phi'$ has a directed edge to another vertex $z_2$ in $\phi'$ only if $z_1$ is to the left of $z_2$ in $\phi'$.  Enumerations $\rho$ and $\rho'$ are symmetric, in the sense that the former is obtained from the latter by exchanging the labels of the vertices.  This is also true for enumerations $\gamma$ and $\gamma'$, so we need to consider only two cases, either $\phi' = \rho$ or $\phi' = \gamma$.  For the sake of simplicity, define $u_1 = w_1$ and $u_2 = v_2$ in $\rho$, and $u_1 = v_2$ and $u_2 = w_1$ in $\gamma$.  Thus, $\phi' = u_1, u_2, v_3, \ldots, v_{n-1}, w_{n-1}$.  Denote by $\ell(v)$ and $r(v)$ the leftmost and rightmost neighbors of $v$ in $\phi$ for every $v \in V(D)$.  The following claims analyze the positions of $w_n$ and $v_n$ in $\phi$.
 \begin{description}
  \item[Claim 1:] $w_n \not \in \cir[u_1, w_{n-1}]$ in $\phi$.  Otherwise, there would be two adjacent vertices $z_1$ and $z_2$ such that $z_1, z_2$ appear in this order in $\phi'$ and $z_1, w_n, z_2$ appear in this order in $\phi$.  But this is impossible, because $z_1 \to z_2$ and $z_1 \nto w_n$.
  \item[Claim 2:] $v_n \in \cir[w_1, w_{n-1}]$ in $\phi$.  Again, suppose otherwise, and observe that $v_n \not \in \cir[w_1, v_3]$ in $\phi$.  Vertices $v_3$ and $v_n$ are adjacent, so either $v_n \to v_3$ or $v_3 \to v_n$.  The former is impossible because it implies that $w_1 \in [v_n, r(v_n)]$, contradicting the fact that $v_n$ is not adjacent to $w_1$.  In the latter case, $v_2 \nto v_n$ because $w_{n-1} \not \in [v_2, r(v_2)]$.  But this is also impossible since $v_3 \in [\ell(v_n), v_n]$, $v_2 \in [v_n, r(v_n)]$, and $v_2 \to v_3$, contradicting the fact that $[\ell(v_n), r(v_n)]$ is out-straight.
 \end{description}
 By Claims 1~and~2, $w_1, v_n, w_{n-1}, w_n$ must appear in this order in a traversal of $\phi$.  So, $v_n \to w_n$ in $D$ because $w_1 \not \in [w_n, r(w_n)]$ or, otherwise, $\phi$ would not be out-round.  Hence, $v_1, w_{n-1}, w_n$ cannot appear in this order in the range $[\ell(v_n), r(v_n)]$, because $v_1$ is adjacent to $w_n$ and not to $w_{n-1}$.  Consequently, $v_n, w_{n-1}, v_1$ appear in this order in $[\ell(v_n), r(v_n)]$.  Analogously, $v_2, w_{n-1}, v_1$ cannot appear in this order in $[\ell(v_n), r(v_n)]$, so $v_n, w_{n-1}, v_2$ must appear in this order in $[\ell(v_n), r(v_n)]$.  Recall that $v_2 \to v_3$ in $D$, thus $v_n, w_{n-1}, v_2, v_3$ must appear in this order in $[v_n, r(v_n)]$ and, therefore, $v_n \to v_3$.  But since $\phi'$ is either $\rho$ or $\gamma$, it follows that $v_n, w_{n-1}, u_1, u_2, v_3$ appear in this order in $\phi$.  This is a contradiction to the fact that $\phi$ is out-round, because $v_n \to v_3$ but $v_n \nto w_1 \in \{u_1, u_2\}$.

 Finally, by using backtracking arguments, it can be proved that there are no locally out-straight orientations of the $3$-sun, the tent, and the umbrella.
\end{proof}

PIG, PCA, and PHCA graphs can also be described in terms of orientations and enumerations.  The analogous of an out-straight oriented graph, for PIG graphs, is called a straight oriented graph.  An oriented graph $D$ is \emph{straight} if there is an ordering $v_1, \ldots, v_n$ of $V(D)$ such that, for every vertex $v_i$,  $N^-[v_i] = \lin[v_{i-\ell}, v_i]$ and $N^+[v_i] =  \lin[v_{i}, v_{i+r}]$, where $\ell = d^-(v_i)$ and $r = d^+(v_i)$ (see Figure \ref{fig:round orientations}~(a)).  As before, the ordering $v_1, \ldots, v_n$ is called a \emph{straight enumeration} of $D$.  \citet{DengHellHuangSJC1996} \citep[see also][]{Huang1992} proved that $G$ is a PIG graph if and only if it admits a straight orientation.  

\begin{figure}
 \begin{tabular}{c@{\ \ }c@{\ \ }c}
  \includegraphics{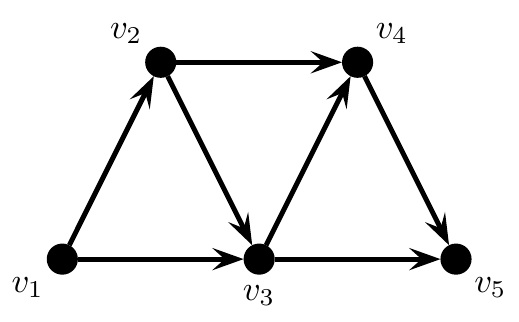} & \includegraphics{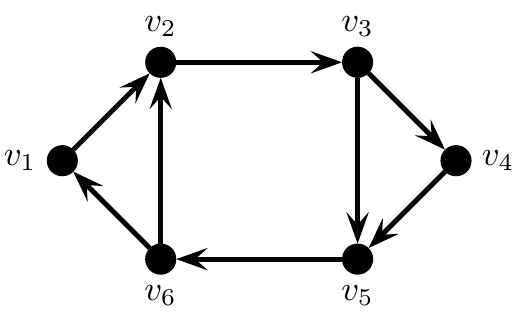} & \includegraphics{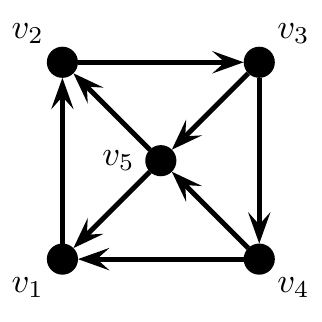} \\
  (a) & (b) & (c)
 \end{tabular}
 \caption[Examples of straight, locally straight and round oriented graphs.]{Examples of straight (a), locally straight (b), and round (c) oriented graphs.}\label{fig:round orientations}
\end{figure}

As out-straight oriented graphs, straight oriented graphs can be generalized either as locally straight oriented graphs or as round oriented graphs. An oriented graph $D$ is \emph{round} if there is an ordering $v_1, \ldots, v_n$ of $V(D)$ such that $N^-[v_i] = \cir[v_{i-\ell}, v_i]$ and $N^+[v_i] = \cir[v_{i}, v_{i+r}]$, for every vertex $v_i$ with $\ell = d^-(v_i)$ and $r = d^+(v_i)$.  When the scope $S$ of each vertex is a straight enumeration of $D[S]$, then $D$ is also a \emph{locally straight} oriented graph.  As usual, the circular orderings corresponding to the round and locally straight oriented graphs are called \emph{round enumerations} and \emph{locally straight enumerations}, respectively.  Examples of straight, locally straight and round oriented graphs are depicted in Figure~\ref{fig:round orientations}. \citet{HellHuangJGT1995} \citep[see also][]{SkrienJGT1982} proved that the class of PCA graphs is exactly the class of graphs that admit a round orientation.  As with NCA graphs, we can restore the Helly condition of PIG graphs by restricting the attention to locally straight graphs.

\begin{theorem}\label{thm:phca iff locally straight orientable}
 A graph is a PHCA graph if and only if it admits a locally straight orientation.
\end{theorem}

\begin{proof}
 Let $\M$ be an NPHCA model of a graph $G$ and define $D$ and $\phi = v_1, \ldots, v_n$ as in the proof of Theorem~\ref{thm:locally-interval}.  By Theorem~\ref{thm:locally-interval}, the scope $S$ of each vertex $v_i$ in $\phi$ is an out-straight enumeration of $D[S]$, for every $1 \leq i \leq n$.  Also, as in the proof of Theorem~\ref{thm:locally-interval}, the submodel $\M'$ of $\M$ induced by the vertices of $S$ is an interval model of $G[S]$.  Since $\M'$ is also PCA then $\M'$ is a PIG model and so $S$ is in fact a straight enumeration of $D[S]$.  Therefore, $D$ is locally straight.

 For the converse, let $G$ be a locally straight orientable graph.  Since locally straight oriented graphs are locally out-straight, then $G$ is an NHCA graph by Theorem~\ref{thm:locally-interval}.  Thus, it is enough to see that $G$ contains no induced $K_{1,3}$ by Theorem~\ref{thm:caracterizacion-NHCA->PHCA}.  But, since the class of locally straight orientable graphs is hereditary, it suffices to prove that no orientation of $K_{1,3}$ is locally straight.  This is clearly true because the $K_{1,3}$ graph admits no round orientations at all.
\end{proof}

In Figure~\ref{fig:round hierarchy}, we depict the class hierarchy of NCA graphs.  This hierarchy uses the same notation as the hierarchy depicted in Figure~\ref{fig:hierarchy}.  This time, however, the label beside the edge between $\C_1$ and $\C_2$ indicates the additional condition that a round enumeration of an orientation of a graph $G \in \C_1$ has to satisfy in order to also prove that $G \in \C_2$.  

\begin{figure}[htb!]
  \centering
  \includegraphics{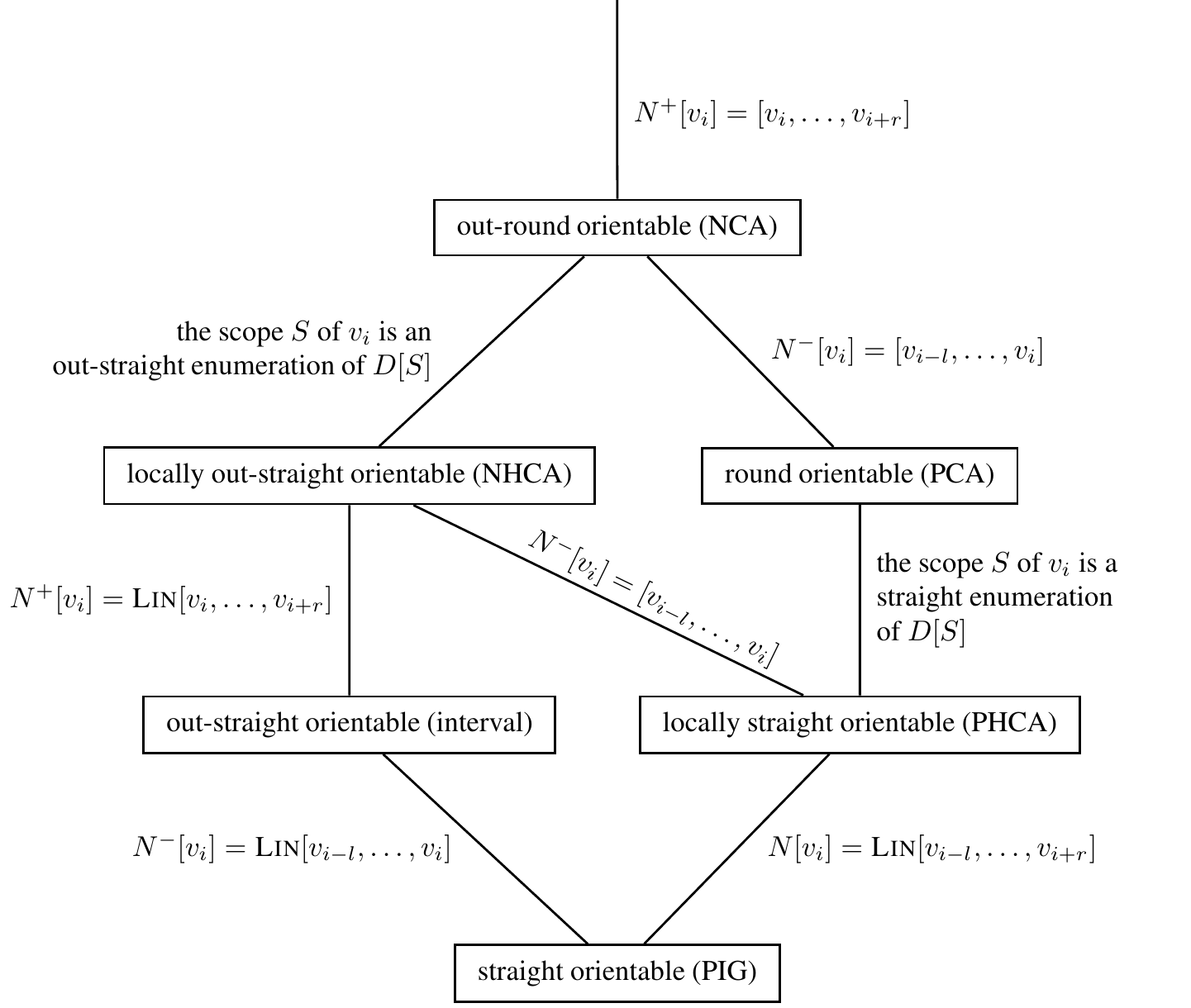}
  \caption{The NCA class hierarchy, depicted in terms of round enumerations.}\label{fig:round hierarchy}
\end{figure}

\section{Conclusions}
\label{sec:conclusions}

In this paper we studied the NHCA, PHCA, and UHCA classes of graphs.  In Section~\ref{sec:motivation} we showed that these classes arise naturally when some properties are generalized from a family of intervals on the real line to a family of arcs of a circle.  NHCA, PHCA, and UHCA also appear when studying some problems on circular-arc graphs that are to hard for the general case but easy for interval graphs.  Given the close relationship between interval and NHCA graphs, we could expect to find good results for such problems, restricted to NHCA or PHCA graphs.  The coloring problem on circular-arc graphs is perhaps the most well known example in which NHCA and PHCA play a major role.  Let $\chi$ and $\omega$ be the chromatic number and clique number of a circular-arc graph $G$, and $r$ be the maximum number of arcs that share a common point of a circular-arc model $\M$ of $G$. \cite{TuckerSJAM1975} proved that $\chi \leq 3/2\omega$ when $G$ is NHCA, while $\chi \leq 3/2r \leq 3/2\omega$ when $\M$ is a PHCA model.  Given the close relationship between interval and NHCA graphs, we expect that many properties and algorithms on (proper) interval graphs can be easily generalized for (PHCA) NHCA graphs.

In Section~\ref{sec:characterizations} we described characterizations for the classes of NHCA, PHCA, and UHCA graphs.  These characterizations imply a complete family of forbidden induced subgraphs for the PHCA and UHCA classes.  That is, adding $W_4$ and the $3$-sun to Tucker's forbidden subgraphs for PCA graphs, we obtain the complete list of forbidden subgraphs for PHCA graphs, while adding $W_4$ to Tucker's forbidden family for UCA graphs, we obtain all the subgraphs forbidden for UHCA graphs.  On the other hand, the characterization for NHCA graphs is partial; we know which circular-arc graphs are not NHCA, but the complete family of forbidden subgraphs for NHCA is still unknown.  

On the algorithmic side, \citet{Nussbaum2008} showed how to transform a circular-arc model into an equivalent PCA model in $O(n)$ time.  Combining this results with those in Section~\ref{sec:algorithms}, we obtain $O(n)$ time algorithms for transforming a circular-arc model into an equivalent PHCA or UHCA model.  For NHCA graphs we did not find such an algorithm.  Furthermore, the $O(n+m)$ time algorithm for the recognition of NHCA graphs is not based in transformations on the input circular-arc model.  The problem is that did not find how to transform an HCA model into an interval model in $O(n)$ time.  

The recognition problem for NHCA, PHCA, and UHCA graphs is well solved when the input is given by the adjacency lists of the graph.  However, the algorithm for the recognition of NHCA graphs given in Section~\ref{sec:algorithms} is hard, because it requires the build of a circular-arc model of the input graph.  So, it would be nice to find a direct algorithm for the recognition of NHCA graphs.  The local interval property could play an important part in such an algorithm.  In fact, many ideas that were successful for other recognition algorithms can perhaps be used for the recognition of NHCA graphs.  For instance, we could test whether $G$ is an interval graph.  If so, output the interval model of $G$.  Otherwise, try to split $G$ into several interval graphs, resembling the process done by \citet{DengHellHuangSJC1996}.  Then, recognize each part as an interval graph and ``rebuild'' the model.  Other possibility is to find the cliques of each part so as to build a $PC$-tree latter \citep[see][]{HsuMcConnellTCS2003}, as \citet{BoothLuekerJCSS1976} do.  The major inconvenient of this approach is that finding the cliques of an NHCA graphs directly seems difficult.  The third possibility is to extend the incremental data structure used by~\cite{KorteMohringSJC1989} in their interval graph recognition algorithm.  Since the insertion of a vertex depends only on its neighborhood, then it is perhaps possible to exploit the local interval property.  Also, it could be interesting to generalize the incremental algorithm by~\cite{Crespelle2010}, so as to insert each vertex in $O(n)$ time.

\end{document}